%% file: RSAd.tex
\documentclass[11pt]{article}
\usepackage[left=1in, right=1in, top=2.5cm]{geometry}

\usepackage{booktabs}

\usepackage{amsmath,amssymb,amsfonts}

\usepackage{algorithmic}
\usepackage{graphicx}

\usepackage{multicol}
\usepackage{wrapfig}
\usepackage{lipsum}

\usepackage{subfigure}
\usepackage{times,hyperref}
\usepackage{epstopdf} 
\usepackage{amsthm}
\usepackage{caption}
\usepackage{color}
\newtheorem{theorem}{Theorem}
\newtheorem{lemma}{Lemma}

\newcommand{\presec}{\vspace{-0.0in}}
\newcommand{\postsec}{\vspace{-0.0in}}

\newcommand{\precaption}{\vspace{-0in}}
\newcommand{\postcaption}{\vspace{-0in}}

\newcommand{\RSG}{\textsc{DyHypes}}
\newcommand{\RSGS}{\textsc{DyHypesServer}}

\usepackage[ruled,vlined,linesnumbered]{algorithm2e}

\begin{document}

\title{\Large Locally Self-Adjusting Hypercubic Networks}
\date{}

\author{Sikder Huq and Sukumar Ghosh \\
Department of Computer Science, The University of Iowa}

\maketitle
\input{abstract}

\vfill\eject
\input{introduction}

\input{model} 
\input{randsg} 
\input{analysis}

\input{rsg_server}

\input{conclusions}

\bibliographystyle{abbrv}
\bibliography{RSAd}
\appendix
\input{appen}

\end{document}

%% file: abstract.tex
\begin{abstract}

In a prior work (ICDCS 2017), we presented a distributed self-adjusting algorithm DSG for skip graphs. DSG performs topological adaption to communication pattern to minimize the average routing costs between communicating nodes. In this work, we present a distributed self-adjusting algorithm (referred to as DyHypes) for topological adaption in hypercubic networks. One of the major differences between hypercubes and skip graphs is that hypercubes are more rigid in structure compared skip graphs. This property makes self-adjustment significantly different in hypercubic networks than skip graphs. Upon a communication between an arbitrary pair of nodes, DyHypes transforms the network to place frequently communicating nodes closer to each other to maximize communication efficiency, and uses randomization in the transformation process to speed up the transformation and reduce message complexity. We show that, as compared to DSG, DyHypes reduces the transformation cost by a factor of $O(\log n)$, where $n$ is the number of nodes involved in the transformation. Moreover, despite achieving faster transformation with lower message complexity, the combined cost (routing and transformation) of DyHypes is at most a $\log \log n$ factor more than that of any algorithm that conforms to the computational model adopted for this work. Similar to DSG, DyHypes is fully decentralized, conforms to the $\mathcal{CONGEST}$ model, and requires $O(\log n)$ bits of memory for each node, where $n$ is the total number of nodes.

\end{abstract}


%% file: introduction.tex

\section{Introduction} \label{sec:introduction} 

Hypercubic networks are widely used in peer-to-peer and parallel computing. The worst-case routing distance in hypercubic networks is $O(\log n)$, where $n$ is the number of nodes in the network. However, given that many real world communication patterns are skewed, topological self-adjustment has the potential to significantly improve the overall routing performance by reducing the average routing distance between frequently communicating nodes.

In this paper, we present a self-adjusting algorithm \RSG for hypercubic networks. Our algorithm performs topological adaptation to unknown communication patterns to maximize communication efficiency. The self-adjusting model we use for hypercubic networks is similar to the one we used for skip graphs \cite{DSG}. In summary, we use $\mathcal{CONGEST}$ model for communications, $O(\log n)$ bit memory for each node, and after any communication $(u,v)$, communicating nodes $u$ and $v$ get attached with a direct link between each other in the transformed network.

Upon a communication request, algorithm \RSG~ performs routing using the standard routing algorithm of hypercubic networks and then partially transforms the network conforming to our self-adjusting model. We show that both the routing and transformation costs of \RSG~ are at most a $\log \log$ factor more than that of the optimal algorithm. 

Compared to DSG, algorithm \RSG~ improves the transformation cost by a logarithmic factor and reduced amortized message complexity by a polylogarithmic factor. A comparison between DSG and \RSG~ is presented in Table \ref{tab:results}. The structure of hypercubic networks is more rigid compared to that of the skip graphs as nodes are split into two \emph{exact} halves as they are placed in the 0-networks and 1-networks in their equivalent tree model. This structural rigidity is a major challenge for self-adjustment in hypercubic networks.

\begin{table}[htbp]
\scriptsize
\centering
\begin{tabular}{|c|c|c|}

\hline
 &\textbf{DSG} & \textbf{\RSG} \\[1ex]
\hline
Routing cost factor to the working set bound & constant   & constant \\
Routing cost factor to the optimal cost & constant   & $O(\log \log n)$ \\
Transformation cost factor to the working set bound & logarithmic   & constant \\
Total cost factor to the working set bound & logarithmic   & constant \\
Total cost factor to the optimal algorithm & $O(\log n)$   & $O(\log \log n)$ \\
Transformation message complexity for routing distance $d$ & $O(d^2 2^d)$   & $O(2^d)$ (amortized)\\
Memory per node & $O(\log n)$ bits   & $O(\log n)$ bits \\
Communication model & $\mathcal{CONGEST}$   & $\mathcal{CONGEST}$ \\
\hline
\end{tabular}
\caption{Results summary}
\label{tab:results}
\end{table}

\subsection{Our Contributions}
\begin{enumerate}
\item We derive a lower bound for the performance of any self-adjusting algorithm designed for hypercubic networks conforming to our self-adjusting model.
\item We propose algorithm \RSG, and show that the cost of our algorithm is at most a $\log \log$ factor more than that of the optimal algorithm.
\item We propose a simple algorithm \RSGS~ for the client-server model.
\end{enumerate}

%% file: model.tex
\presec
\section{Model and Definitions} \label{sec:model} 
\postsec

Each node of an $N$-dimensional hypercube has an $N$-bit coordinate. We denote the $N$-bit coordinate of node $x$ as $Coord(x)$. We also refer to the $i$-th bit of the coordinate of node $x$ as $Coord_i(x)$, where $1 \leq i \leq \log N$. Obviously, $N = \log n$, where $n$ is the total number of nodes.

We use a binary tree based representation \footnote{A similar tree representation was used in our prior work \cite{DSG} to represent a skip graph, where each node of the tree maps to a linked list in the Skip Graph.} of hypercubic networks. We represent an $N$-dimensional hypercube by a full binary tree of height $N$, where each node of the tree represents a $k$-dimensional ($k \leq N$) hypercube within the $N$-dimensional hypercube. The root node represents the entire hypercube, and each child of the root node represents an $(N-1)$-dimensional hypercube (subgraph of the entire $N$-dimensional hypercube). For any non-leaf node in the tree representation, we refer to one of its children as the 0-subnetwork and the other as the 1-subnetwork. Any node $x$ in the $(n-1)$-dimensional hypercubes represented by the 0-subnetwork and 1-subnetwork have the first bit of their coordinate ($Coord_1(x)$) as 0 and 1, respectively.

Similarly, a ``grand-child'' of the root node represents an $(N-2)$-dimensional hypercube, and tree is constructed recursively by splitting the hypercube represented by any non-leaf node into 0 and 1 subnetworks, based on the coordinates of the nodes. We say the root node of the tree is at level 0, and any node of the tree at level $i$ represents an $(N-i)$-dimensional hypercube. Also, any node $x$ in $b$-subnetwork at level $i$ has $Coord_i(x) = b$, where $b \in \{0,1\}$. Figure \ref{fig:hypercube_tree} shows an example of the tree representation of a 3-dimensional hypercube.

\begin{figure}[htb]
\def \subfigcapskip{0pt}
\hspace{-0.2in}
\centering
\subfigure[Tree representation of a hypercubic network of 8 nodes.]
    {\label{fig:hypercube_tree} \includegraphics[width=0.55\columnwidth]{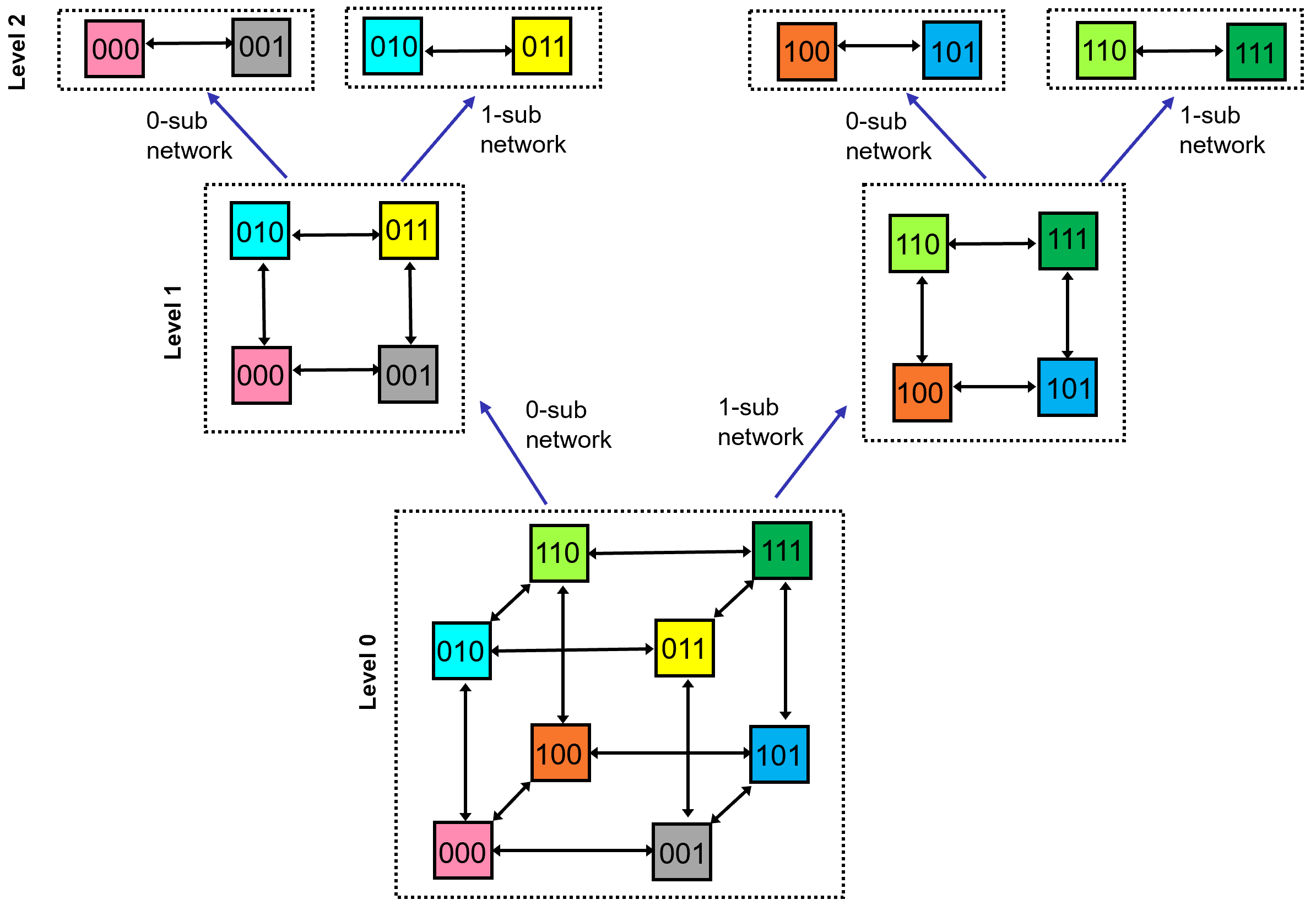}}
\centering
\subfigure[Tree representation with groups.]
    {\label{fig:dyhyp_ex_groups} \includegraphics[width=0.5\textwidth]{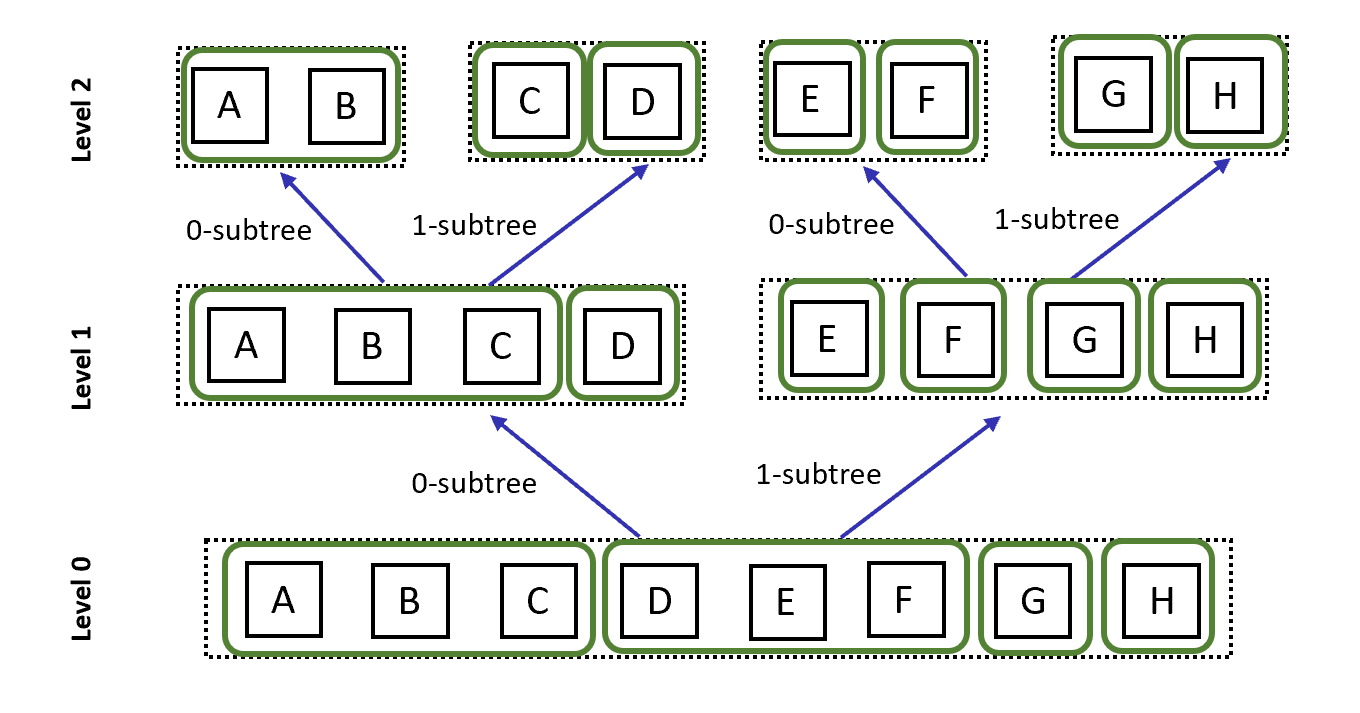}}
\precaption
\caption{Figure in (a) shows how a hypercubic network is representation by a binary tree based on the coordinates of the nodes. Figure in (b) shows an example of groups formed by the nodes. The green boxes indicate groups at different levels. Nodes are placed in the incremental order of their coordinates from left to right.}
\postcaption
\end{figure}

%

\noindent
\textbf{Definition (\textit{Level-$d$ subtree}).} For any node $x$, the level-$d$ subtree of node $x$ is the $(N-d)$-hypercube $C_{N-d}$, such that $x \in C_{N-d}$ and $C_{N-d}$ is represented by a node in the tree representation of the entire $N$-dimensional hypercube.

For example, in Figure \ref{fig:hypercube_tree}, the level-1 subtree of the node with coordinate 010 is the 2-dimensional hypercube consisting of nodes with coordinates 000, 001, 010 and 011.

\noindent
\textbf{Definition (\textit{Level-d complementary subtree}).} An $(N-d)$-hypercube $C_{N-d}$ is a level-$d$ complementary subtree of a node $x$ if and only if (a) $C_{N-d}$ is not the level-$d$ subtree of node $x$, and (b) $C_{N-d}$ is a subgraph of the level-$(d-1)$ subtree of node $x$ and.

For example, in Figure \ref{fig:hypercube_tree}, the level-1 complementary subtree of the node with coordinate 010 is the 2-dimensional hypercube consisting of nodes with coordinates 100, 101, 110 and 111.

We denote the level-$d$ subtree and level-$d$ complementary subtree of node $x$ as $s^x_d$ and $\sim s^x_d$, respectively.

For a pair of communicating nodes $(u,v)$ in a $N$-dimensional hypercube, let $L_{lca}(u,v)$ be the highest level in the tree representation at which there is a node mapped to a subtree containing both nodes $u$ and $v$. For example, in the hypercube shown in Figure \ref{fig:hypercube_tree}, if $u$ is the node 000 and $v$ is the node 010, then $L_{lca}(u,v) = 1$.

\noindent
\textbf{Definition (\textit{Tree Distance}).} The tree distance between nodes $u$ and $v$ is $N - L_{lca}(u,v)$.

For example, in the 3-dimensional hypercube shown in Figure \ref{fig:hypercube_tree} the tree distance between nodes 000 and 111 is 3, where the actual shortest distance is 1. We denote the tree distance of nodes $u$ and $v$ in the hypercube $\mathcal{N}_t$ as $d_{Tree}(\mathcal{N}_t, (u,v))$.

\noindent
\subsection{Computational and Self-Adjusting Model.} We consider a similar computational and self-adjusting model that we used in \cite{DSG}. We assume a synchronous message passing model, where communications occur in \emph{rounds}. Our computational model has the following properties:

\begin{enumerate}
    \item A node can send and receive at most 1 message through a link in a round.
    \item The size of each message is $O (\log n)$ bits (i.e. $\mathcal{CONGEST}$ model).
    \item Each node holds a memory of size $O (\log n)$ bits.

\end{enumerate}

Let $V$ be a set of given $2^N$ nodes and $\mathcal{N}$ be the family of all possible $N$-dimensional hypercubes formed by the nodes of $V$. Let $\sigma = (\sigma_1, \sigma_2, ... , \sigma_{m})$ be an unknown access sequence consisting of $m$ sequential communication requests, $\sigma_t = (u,v) \in V \times V, u \neq v$ denotes a routing request from source $u$ to destination $v$ at time $t$. Given a hypercube $\mathcal{N}_t$, we define the routing distance $d(\mathcal{N}_t, \sigma_t)$ as the number of intermediate nodes in the shortest path between the source and destination associated with request $\sigma_t$.

At any time $t$, let $\mathcal{N}_t \in \mathcal{N}$ be a hypercube and a pair of nodes $(u,v) \in V \times V, u \neq v$ communicate, a self-adjusting algorithm performs the followings:

\begin{enumerate}
\item Establishes communication between nodes $u$ and $v$ in $\mathcal{N}_t$.
\item Transforms the network $\mathcal{N}_t$ to another network $\mathcal{N}' \in \mathcal{N}$, such that nodes $u$ and $v$ move to a subtree of size two in $\mathcal{N}'$. This implies that a nodes $u$ and $v$ get connected by a direct link after the transformation.
\end{enumerate}

Let an algorithm $\mathcal{A}$ transforms the hypercube $\mathcal{N}_t$ to $\mathcal{N}_{t+1}$. We define the cost for network transformation as the number of rounds needed to transform the topology. We denote this transformation cost at time $t$ as $\rho( \mathcal{A}, \mathcal{N}_t, \sigma_{t})$. Similar to the prior work \cite{splayNet,DSG}, we define the cost of serving request $\sigma_{t}$ as the routing distance between the communicating nodes plus the cost of transformation performed by $\mathcal{A}$ plus one, \emph{i.e.}, $d(\mathcal{N}_t, \sigma_t)+ \rho( \mathcal{A}, \mathcal{N}_t, \sigma_{t})+1$.

\noindent
\textbf{Definition (\textit{Average and Amortized Cost}).} Given an initial hypercube $\mathcal{N}_0$, the average cost for algorithm $\mathcal{A}$ to serve a sequence of communication requests
$\sigma = (\sigma_1, \sigma_2, \cdots, \sigma_m)$ is:

\begin{equation}
Cost( \mathcal{A}, N_0, \sigma) = \frac{1}{m} \sum_{i=1}^{m} (d(\mathcal{N}_i, \sigma_i)+\rho( \mathcal{A}, \mathcal{N}_i, \sigma_{i})+1)
\end{equation}

The amortized cost of $\mathcal{A}$ is defined as the worst case cost to serve a
communication sequence $\sigma$, \textit{i.e.} $max_{N_0,\sigma}$ Cost$(\mathcal{A}, N_0, \sigma)$.

\section{Working Set Property} \label{sec:ws}
\noindent

\textbf{Definition (\textit{Communication Graph.}).} Given a communication network and a set of communications between arbitrary pair of nodes, A \emph{communication graph} $G(V,E)$ is a simple undirected graph where $V$ is the set of all the nodes of the network, and $E$ represents all the communications as each edge $e \in E$ connects a pair of communicating nodes.

In this paper, we use the notation $\mathcal{G}_x(t^\prime,t)$ to denote the connected component containing node $x$ in the communication graph drawn for the communications that took place between the time interval starting from time $t^\prime$ (inclusive) and ending at time $t$ (exclusive). Also, if $t$ is referred to as the current time, we often write $\mathcal{G}_x(t^\prime)$ instead of $\mathcal{G}_x(t^\prime,t)$.

\textbf{Definition (\textit{Working Set Number}).} Then the \emph{working set number} for nodes $u$ and $v$ at time $t$ is defined as follows:

\begin{itemize}
\item \textbf{If nodes $u$ and $v$ communicated earlier:} Let $t^\prime$ be the last time nodes $u$ and $v$ communicated. The working set number for nodes $u$ and $v$ at time $t$ is the number of nodes in the connected component in $\mathcal{G}_u(t^\prime,t)$.

\item \textbf{If nodes $u$ and $v$ never communicated earlier:} Let $\mathcal{G}_u(0,t)$ be the connected component of node $u$ in the communication graph drawn for all the communications until time $t$ (exclusive). If $v$ is a node in $\mathcal{G}_u(0,t)$, the working set number for nodes $u$ and $v$ at time $t$ is the number of nodes in $\mathcal{G}_u(0,t)$. Otherwise (if $v$ is not a node in $\mathcal{G}_u(0,t)$), the working set number is $\max (2^d, |V_u| + |V_v|)$, where $d$ is the tree distance between nodes $u$ and $v$ in the hypercube at time $t$, and $V_u$ and $V_v$ are the set of vertices in $\mathcal{G}_u(0,t)$ and $\mathcal{G}_v(0,t)$, respectively.

\end{itemize}

As an example, for the latest communication request $(u,v)$ shown in figure \ref{fig:wset}(a),
the corresponding communication graph $G$ is shown in figure \ref{fig:wset}(b). The number of distinct
nodes in $G$ that have a path from either $u$ or $v$ is 5; therefore the working set number for
the communication request is 5.

We denote the working set number for node pair $(u,v)$ at time $t$ as $T_t(u,v)$.

\noindent
\textbf{Definition (\textit{Working Set Property}).} For a hypercube $\mathcal{N}_i$ at time $i$, the
\emph{working set property} for any node pair $(x,y)$ holds if and only if $d_{Tree}(\mathcal{N}_t, (x,y)) \leq \left \lceil \log_2 T_i(x,y) \right \rceil$.

\noindent
\textbf{Definition (\textit{Working Set Bound}).} We define the \emph{working set bound} as
$WS(\sigma) = \sum_{i=1}^{m} \left \lceil \log_2( T_i(\sigma_i)) \right \rceil$.

\begin{figure}[!t]
\def \subfigcapskip{0pt}
    \centering
    \subfigure[An access pattern showing a repeating communications between $u$ and $v$.]
    {\includegraphics[width=0.45\columnwidth]{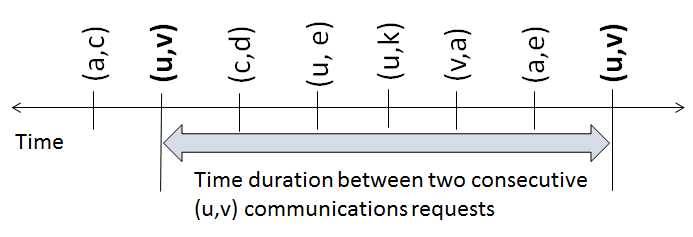}}
    \centering
    \hspace{0.1in}
    \subfigure[Communication graph $G$ for the time duration shown in (a).]
    {\includegraphics[width=0.35\textwidth]{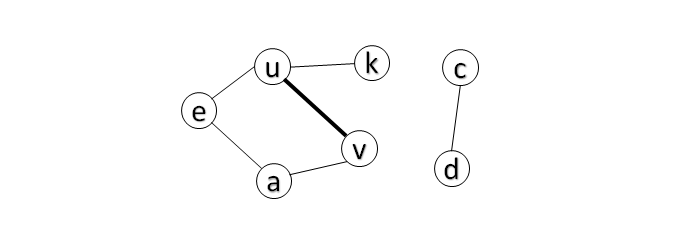}}
    \caption{For the access pattern shown in (a), the working set number for the last communication between $u$ and $v$ is 5, as the the number of distinct nodes in $G$ that has a path from either $u$ or $v$ is 5 (e,a,k,u and v).}
    \label{fig:wset}
\end{figure}

\subsection{Correlation with Skip Graphs} \label{sec:ws_sig}
\noindent

In our prior work \cite{DSG}, we showed that there exists a communication sequence for which the routing cost for any algorithm designed for skip graphs and conforming to our model is at least $WS(\sigma)$ rounds. We used this result to show that the worst case routing cost incurred by our algorithm for skip graphs is at most a constant factor more than that of the optimal algorithm. Although our computational and self-adjusting model is similar for hypercubes, we cannot directly use this result for hypercubes primarily because the structural difference between skip graphs and hypercubes, as described below.

For skip graphs, the self adjusting model requires that the communicating nodes move to a linked list of size 2 upon routing. One of the key differences between skip graphs and hypercubes is that, a node is a part of exactly one linked list of size 2 in a skip graph, whereas a node is a part of $N$ different 2-dimensional hypercubes in an $N$-dimensional hypercube. For example, in the 3-dimensional hypercube in Figure \ref{fig:hypercube_tree}, the node 000 is a part of three 2-dimensional hypercubes (i.e. (000,001), (000, 010) and (000,110)). However, observe that only one 2-dimensional hypercube consisting of node 000 is mapped by a node in the binary tree at level 2.

\begin{lemma}
\label{lemma:ws}
\textbf{(Tree Distance Lemma)} For any node $u$ in an $N$-dimensional hypercube $\mathcal{N}_t$ at any time $t$, there exists a node $v$ in $\mathcal{N}_t$ such that $d_{Tree}(\mathcal{N}_t, (u,v)) \geq \left \lceil \log_2 T_t(u,v) \right \rceil$.

\end{lemma}

\begin{proof} Let $\mathcal{G}_u(0,t)$ be the connected component of node $u$ in the communication graph drawn for all the communications until time $t$ (exclusive), and $V_u$ be the set of vertices in $\mathcal{G}_u(0,t)$. If $|V_u| = 1$, then $u$ never communicated and the lemma holds trivially. Let us assume $|V_u| > 1$.  We argue that there exists a node $v \in V_u$ such that $d_{Tree}(\mathcal{N}_t, (u,v)) \geq \left \lceil \log_2 T_t(u,v) \right \rceil$.

Obviously, there exists a $k$, such that $1 \leq k \leq N$ and $2^{k-1} < |V_u| \leq 2^k$. This implies that there is a non-empty set of nodes $S_x$, such that $\forall x \in S_x : 2^{k-1} < T_t(u,x) \leq 2^k$, and $|S_x| \geq  |V_u| - 2^{k-1}$. Now, if $\exists x \in S_x : d_{Tree}(\mathcal{N}_t, (u,x)) \geq k$, then the lemma holds immediately. Otherwise, $\forall x \in S_x : d_{Tree}(\mathcal{N}_t, (u,x)) < k$, and  $|V_u| - 2^{k-1}$ nodes with working set number less than or equal to $2^{k-1}$ must have a tree distance $k$ or more. Hence the lemma follows.
\end{proof}

\begin{lemma}
\label{lemma:tree_distance}
\textbf{(Working Set Lemma)} For any node $u$ in an $N$-dimensional hypercube $\mathcal{N}_t$ at any time $t$, there exists a node $v$ in $\mathcal{N}_t$ such that $d(\mathcal{N}_t, (u,v)) \geq \frac{\left \lceil \log_2 T_t(u,v) \right \rceil}{O(\log \log n)}$, where $n = 2^N$.

\end{lemma}

\begin{proof}
We know that in hypercube the distance between nodes $u$ and $v$ is $k$ if and only if the \emph{Hamming distance} between their coordinates is $k$. Thus, in a $N$-dimensional hypercube, the number of nodes with distance $k$ from any node is $N$ choose $k$. Note that, $N = \log_2 n$. Using the upper bound of binomial coefficients and rule of logarithms, we get:

\begin{equation}
\label{eq:hyp_opt_thm}
\begin{split}
\binom{N}{k} = & \binom{\log_2 n}{k}  \leq  \Bigg( \frac{e\cdot \log_2 n}{k} \Bigg)^k = 2^{k \cdot \log_2 \big( \frac{e\cdot \log_2 n}{k} \big)} = 2^{k \log_2 \big( \frac{e}{k} \big) + k \log_2 \big(\log_2 n \big)} = 2^{O(k \cdot \log \log n)}
\end{split}
\end{equation}

Thus, the number of nodes $x$ in $\mathcal{N}_t$ such that $d(\mathcal{N}_t, (u,x)) \leq k$ and $x \neq u$ is $\sum_{i=1}^{k} \binom{N}{i}$. From \cite{binom_sum}, we know,

\begin{equation}
\label{eq:binom_sum}
\begin{split}
\sum_{i=0}^{k} \binom{N}{i} \leq \binom{N}{k} \frac{N-(k-1)}{N-(2k-1)}
\end{split}
\end{equation}

If $k \leq \frac{N}{4}$,

\begin{equation}
\label{eq:binom_sum_coef}
\begin{split}
\frac{N-(k-1)}{N-(2k-1)} < \frac{N-((N/4)-1)}{N-(2(N/4)-1)} < \frac{N-(N/4)}{N-(N/2)} = \frac{3}{2}
\end{split}
\end{equation}

Using Equations \ref{eq:binom_sum}, \ref{eq:binom_sum_coef} and \ref{eq:hyp_opt_thm}, when $k \leq \frac{N}{4}$:

\begin{equation}
\label{eq:binom_sum_2}
\begin{split}
\sum_{i=0}^{k} \binom{N}{i} \leq \frac{3}{2}\binom{N}{k} = 2^{O(k \cdot \log \log n)}
\end{split}
\end{equation}

Also, if $k > \frac{N}{4}$, clearly, $\sum_{i=0}^{k} \binom{N}{i} \leq 2^{4k}$, since $n=2^N$. Thus, $\sum_{i=1}^{k} \binom{N}{i} = 2^{O(k \cdot \log \log n)}$ for any $k$, $1 \leq k \leq N$.

From Lemma \ref{lemma:ws}, we know that there exists a node $v$ in $\mathcal{N}_t$ such that $d_{Tree}(\mathcal{N}_t, (u,v)) \geq \left \lceil \log_2 T_t(u,v) \right \rceil$. Let $S_u$ be the set of nodes in $\mathcal{N}_t$ such that $\forall x \in S_u: \left \lceil \log_2 T_t(u,x) \right \rceil \leq \left \lceil \log_2 T_t(u,v) \right \rceil$. Let $d(\mathcal{N}_t, (u,v)) = d+1$. Now, if $\forall x \in S_u: d(\mathcal{N}_t, (u,x)) \leq d(\mathcal{N}_t, (u,v))$, then

\begin{equation}
\label{eq:binom_sum_3}
\begin{split}
\sum_{i=0}^{d} \binom{N}{i} = 2^{\left \lceil \log_2 T_t(u,v) \right \rceil - 1}
\end{split}
\end{equation}

From Equations \ref{eq:binom_sum_2} and \ref{eq:binom_sum_3}, we get:

\begin{equation}
\label{eq:binom_sum_4}
\begin{split}
2^{O(d \cdot \log \log n)} = 2^{\left \lceil \log_2 T_t(u,v) \right \rceil - 1}
\end{split}
\end{equation}

Hence, $d = \frac{\left \lceil \log_2 T_t(u,v) \right \rceil}{O(\log \log n)}$. However, if $\exists x \in S_u: \left \lceil \log_2 T_t(u,x) \right \rceil > \left \lceil \log_2 T_t(u,v) \right \rceil$,
similar argument can be used to show that $d(\mathcal{N}_t, (u,x)) \geq \frac{\left \lceil \log_2 T_t(u,v) \right \rceil}{O(\log \log n)}$.
\end{proof}

\begin{theorem}
\label{theorem:hypercube_ws}
\textbf{(Hypercube Optimality Theorem)} There exists a communication sequence $\sigma = ( \sigma_1, \sigma_2, \cdots \sigma_m)$ such that the routing cost of any self-adjusting algorithm conforming to our model for a $(\log n)$-dimensional hypercube is at least $\frac{WS(\sigma)}{O(\log \log n)}$ rounds.

\end{theorem}

\begin{proof}
We start with any communication sequence $\sigma^\prime = ( \sigma_1^\prime, \sigma_2^\prime, \cdots \sigma_m^\prime)$, where $\sigma_i^\prime = (u, v)$. We construct $\sigma$ from $\sigma^\prime$ as follows. For any $i$, if $d(\mathcal{N}_t, (u, v)) \geq \frac{\left \lceil \log_2 T_t(u,v) \right \rceil}{O(\log \log n)}$, then we set $\sigma_i = (u,v)$. Otherwise, we know from Lemma \ref{lemma:tree_distance} that there exits a node $v^\prime$ such that $d(\mathcal{N}_t, (u, v^\prime)) \geq \frac{\left \lceil \log_2 T_t(u,v) \right \rceil}{O(\log \log n)}$. We set $\sigma_i = (u,v^\prime)$.

By construction, $\sum_{i=1}^{m} d(\mathcal{N}_t, \sigma_i) \geq \frac{WS(\sigma)}{O(\log \log n)}$.

\end{proof}

%% file: randsg.tex
\presec
\section{Proposed Algorithm} \label{sec:dyhyp}
\postsec

\subsection{Algorithm Overview and Notations}
We propose a randomized and distributed self-adjusting algorithm \RSG~ for hypercubic networks to perform topological adaptation to unknown communication patterns. Upon a communication request, our algorithm \RSG~ first establishes communication using the standard routing algorithm of the network, and then partially transforms the topology according to our self-adjusting model. we refer to the hypercubic network at time as $\mathcal{N}_t$.

Let $N$ be the dimension of the given hypercubic network, and hence $N$ is the height of the tree modeling of the network. Our algorithm \RSG~ allows nodes to form groups at different levels to keep frequently communicating nodes together in the network.


\textbf{Definition (Group).} At any time $t$, a set of nodes $S$ are considered to be in a group at level $d$ if the following are true:

\begin{enumerate}
\item There exist a start coordinate $\mathcal{S}_d$ and an end coordinate $\mathcal{E}_d$ such that $\mathcal{E}_d \geq \mathcal{S}_d$, and the number of nodes in $S$ is exactly $\mathcal{E}_d - \mathcal{S}_d + 1$. All the nodes in $S$ are positioned together in the positions ranged by coordinates $[\mathcal{S}_d,\mathcal{E}_d]$ in network $\mathcal{N}_t$.

\item All the nodes in $S$ are placed in the same level-$d$ subtree.



\end{enumerate}

It is possible that nodes of a group at a level are split into multiple groups at upper levels. Figure \ref{fig:dyhyp_ex_groups} shows an example of groups at different levels in a hypercubic network.

Our algorithm \RSG~ maps the groups linearly in the network. In other words, we consider the node with the lowest coordinate as the left most node of the network, and the node with the highest coordinate as the right most node of the network, and all other nodes are ordered from left to right in the ascending order of their coordinates. As the coordinates of nodes change during transformation when nodes change their position, our algorithm ensures that all nodes of a group are placed in the network adjacently. This implies that each group has a start and an end coordinates (the coordinates of the left-most and right-most nodes of the group). For a group at level $d$, the start and end coordinates are denoted as $S^x_d$ and $E^x_d$, respectively,  where $x$ is any node in the group.

\textbf{Definition ($d$-Related groups).} Let $g$ be a group at a level $d+1$ such that $g$ contains a node $x$ and all nodes of $g$ are positioned in the subtree $s^x_{d+1}$ in network $\mathcal{N}_t$. Let $t^\prime$ be a time such that $t^\prime > t$. If group $g$ is broken into two groups $g_1$ and $g_2$ in network $\mathcal{N}_t^\prime$, such that all nodes of $g_1$ and $g_2$ are positioned in subtrees $s^x_{d+1}$ and $\sim s^x_{d+1}$, respectively, then groups $g_1$ and $g_2$ are said to be $d$-related groups.

Also, if two groups are $d$-related, we say they are $d$-relatives of each other.


\subsection{Major Algorithmic Challenges}
\label{sec:dyhyp_challenges}

There are few algorithmic challenges that we need to address. First, when two nodes from two different groups communicate, we merge the communicating groups. While merging, we need to move other groups to make room for the merging groups. This may result in violation of working set property for any non-communicating groups moved due to the transformation.

Second, it is possible that nodes of multiple groups belong to the same connected component in the communication graph drawn for all communications until the time of current communication. This is especially true when multiple groups are split from the same group at a lower level, as shown in Figure \ref{fig:dyhyp_ex_groups}. Thus, while moving one of such groups to a new position in the network, a challenge is not to violate the working set property for the nodes across the groups.

Third, in hypercubic networks, nodes of a subtree at any level are split into two equally sized subtrees at the immediate upper level. While merging two groups, it is possible that the merged group is too big to fit in a subtree. In that case we need to split the merged group, which may result in violation of the working set property of the nodes of the split groups.

\subsection{Our Key Idea to Address the Challenges}

Our algorithm \RSG~ form groups at different levels to keep frequently communicating nodes together. Due to the structure of hypercubes, a large group at a lower level maybe split into multiple groups at upper levels. Each node has two timestamps for each level, and these timestamps are used to indicate a node's attachment to its group. We use these timestamps to decide how to rearrange (e.g. merge, split, swap etc.) the groups as a part of a transformation.

Our algorithm also focuses on minimizing the number of $d$-relative groups for all levels $d$ at any point of time. Note that, due to the structural rigidity of hypercubic networks, it is impossible to entirely eliminate the existence of the relative groups. For an example, there if a network of 8 nodes has three groups of size 3,3 and 2, there is no way these nodes can be placed in a 3-dimensional hypercube without violating the working set property.

To this end, \RSG~ picks randomly chosen groups and split them into relative groups and at the same time keeps the number of relative groups in the network low. This forces the ``adversary" to try communications within different non-split groups to increase the likelihood of finding a communication between separated relative groups. We refer to the communications across the relative groups as ``bad communications" and any other communications as ``good communications".

Our algorithm \RSG~ charges any communication more than their working set numbers. Thus, ``good communications" are overcharged to pay for the ``bad communications". In other words, algorithm \RSG~ ensures that, the expected number of ``good communications" are large compared to that of the ``bad communications".

To keep number of relative groups low, we maintain the following invariant at all time:

\textbf{Invariant I:} At any time $t$ and in any subtree $s_d$ at any level $d$ in network $\mathcal{N}_t$, there are at most two groups that are $d$-relatives of each other.

To keep track of relative groups, each node holds information about the pair of relative groups for each level. More specifically, for any level $d$, a node $x$ needs to hold four variables: the coordinates of the start (left most) and end (right most) nodes of the group from $d$-relative groups (if exists) in the subtree $s^x_{d+1}$, and the coordinates of the start and end nodes of the other $d$-relative group in the subtree $\sim s^x_{d+1}$. We denote these variables as $\mathcal{S}^x_d$, $\mathcal{E}^x_d$, $\sim \mathcal{S}^x_d$, and $\sim \mathcal{E}^x_d$, respectively. These variables may set to \texttt{NULL} if the corresponding pair of relative groups do not exist.

To address the third challenge listed in Section \ref{sec:dyhyp_challenges}, algorithm \RSG~ uses a randomized approach. When there is a need to evict a subgroup (i.e., a group inside a communicating group) from a subtree to make room for a communicating group, \RSG~ chooses one or more subgroups randomly for eviction. We show that, with this approach, the expected distance between any two nodes remain at most constant factor more than the logarithm of the working set number of the nodes.

\subsection{The Algorithm: Dynamic Hypercubic Networks (\RSG~)}

Upon a communication request and routing, communicating nodes $(u,v)$ record the level of the smallest common subtree that contains both nodes $u$ and $v$. Let the recorded level be $\alpha$. Our algorithm \RSG~ performs transformation conforming to the self-adjusting model to move the communicating nodes in a subtree of size 2. Transformation may take up to three steps depending upon the existing network topology. These three steps are described in the following subsections.

\subsubsection{Subtree Leap} \label{sec:leap}

Let $l(u)$ be the lowest level such that a pair of $l(u)$-relative groups exist in subtree $s^u_{l(u)}$. Similarly, let $l(v)$ be the lowest level such that a pair of $l(v)$-relative groups exist in subtree $s^v_{l(v)}$. This step is executed only if $\alpha$ happens to be lower than both $l(u)$ and $l(v)$. In this step the subtree $s^v_{\min(l(u),l(v))}$ is swapped/replaced by the subtree $\sim s^u_{\min(l(u),l(v))}$. As a result, subtrees $s^u_{\min(l(u),l(v))}$ and $s^v_{\min(l(u),l(v))}$ complement each other. Note that, this will not move any group away from its relative groups.

\subsubsection{Inter-Group Transformation} \label{sec:inter-group}
In this step, we bring the communicating groups adjacent to each other. This step is necessary only if the communicating nodes are from two different groups at level $\alpha$. Let, $g(u)$ be the highest level such that the node $u$'s group at level $g(u)$ is not a strict subset of any other group at any of the levels lower than $g(u)$.
Let $g(v)$ is the similar highest level for node $v$. Let $\alpha$ be the level of the smallest common subtree that contains both nodes $u$ and $v$ after subtree leap. If the communicating nodes belong to different groups at level $\alpha$ (meaning, the group ids $G^u_\alpha$ and $G^v_\alpha$ are different) and either $g(u)$ or $g(v)$ is greater than $\alpha$, this step is executed and as a result node $u$'s and $v$'s groups move next to each other to merge together.

Our algorithm \RSG~ ensures that the Invariant I is preserved after this transformation. If there is any $\alpha$-relative groups in the current network, both nodes $u$ and $v$ are aware of their existence as they hold the information in variables $\mathcal{S}^u_\alpha$, $\sim \mathcal{S}^u_\alpha$ etc. Suppose $A$ and $B$ are the groups of nodes $u$ and $v$ at level $\alpha$, and $C_1$ and $C_2$ are the $\alpha$-relative groups such that groups $A$ and $C_1$ are in subtree $s^u_{\alpha+1}$ and groups $B$ and $C_2$ are in subtree $s^v_{\alpha+1}$, as shown in Figures \ref{fig:ex_intergroup}(a) and \ref{fig:ex_intergroup_random}(a).

Let group $A$ be larger than group $B$ in size (this can be computed using the variables $S^u_\alpha$, $E^u_\alpha$, $S^v_\alpha$ and $E^v_\alpha$), so we move group $B$ next to group $A$ as shown in Figures \ref{fig:ex_intergroup}(c) and \ref{fig:ex_intergroup_random}(c). If the size of group $B$ is larger than the size of group $C_1$, $B$ can be swapped with a portion of $C_1$ as shown in Figure \ref{fig:ex_intergroup}(b). However, if the size of $B$ is smaller than the size of $C_1$, we randomly choose a coordinate $m_r$ in subtree $s^u_{\alpha+1}$ excluding the range of coordinates assigned for the groups $A$ and $C_1$, as shown in Figure \ref{fig:ex_intergroup_random}(a). We take a portion of the network $R$ at the position $m_r$ such that the sum of the size of group $A$ and the size of portion $R$ is greater than the size of group $B$. Then we move group $B$ next to group $A$ as shown in Figure \ref{fig:ex_intergroup_random}(c).

\begin{figure}[htbp]
\def \subfigcapskip{0pt}
\hspace{-0.2in}
\centering
    \subfigure[Groups $A$ and $B$ are communicating and groups $C_1$ and $C_2$ are $\alpha$-relative groups.]
    {\includegraphics[width=0.45\columnwidth]{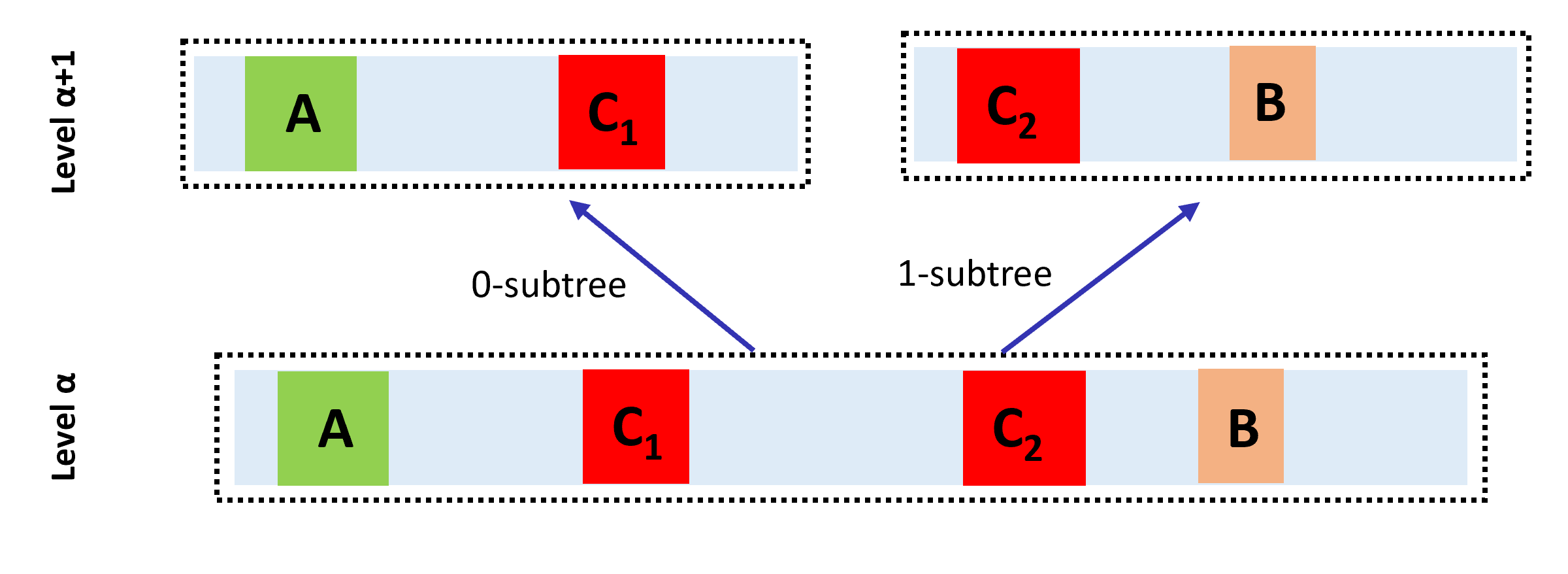}}
\hspace{0.2in}
\centering
    \subfigure[Groups $A$ and $C_1$ move next to each other without violating Invariant I. Similarly groups $B$ and $C_2$ move next to each other in the 1-subtree at level $\alpha + 1$.]
    {\includegraphics[width=0.45\columnwidth]{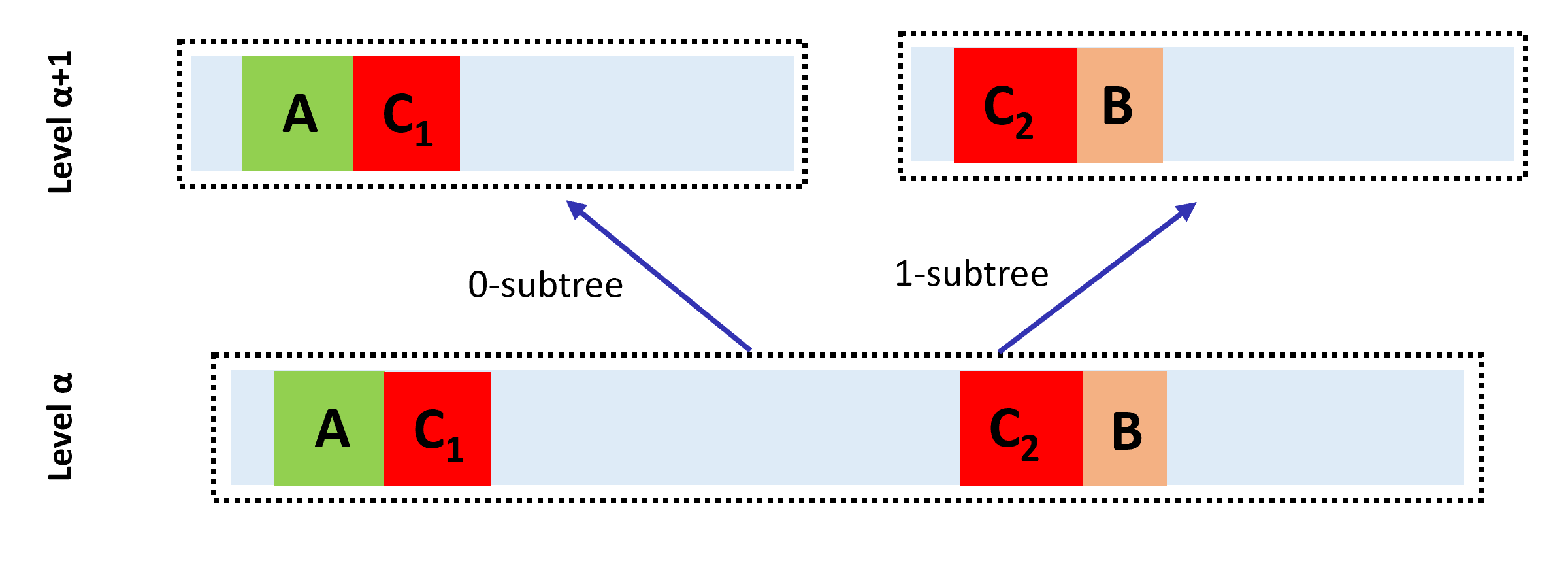}}
\hspace{0.2in}
\centering
    \subfigure[Groups $A$ and $B$ move next to each other without violating Invariant I.]
    {\includegraphics[width=0.45\columnwidth]{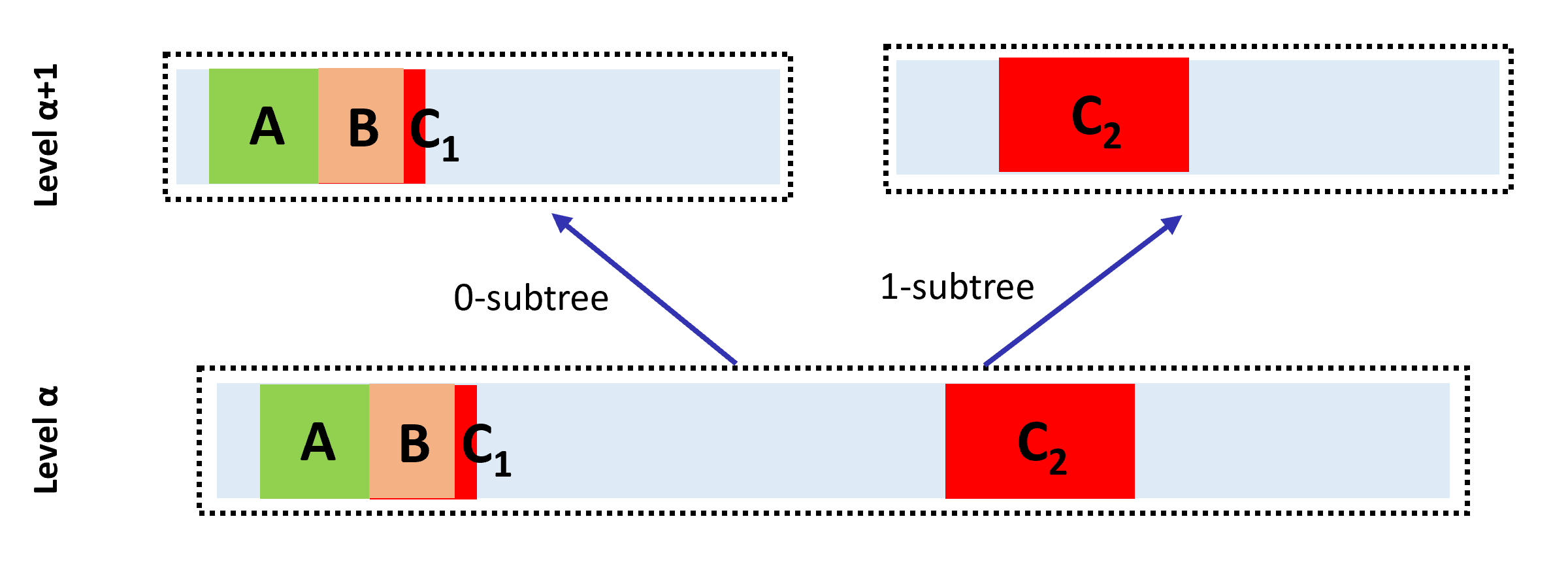}}
\caption{Intergroup transformation example: group $B$ moves next to group $A$, and the size of group $B$ is larger than the size of the $\alpha$-relative group $C_1$.}
\label{fig:ex_intergroup}
\end{figure}

\begin{figure}[htbp]
\hspace{-0.2in}
\def \subfigcapskip{0pt}
\centering
    \subfigure[Groups $A$ and $B$ are communicating and groups $C_1$ and $C_2$ are $\alpha$-relative groups. $R$ is the randomly chosen portion of the network.]
    {\includegraphics[width=0.45\columnwidth]{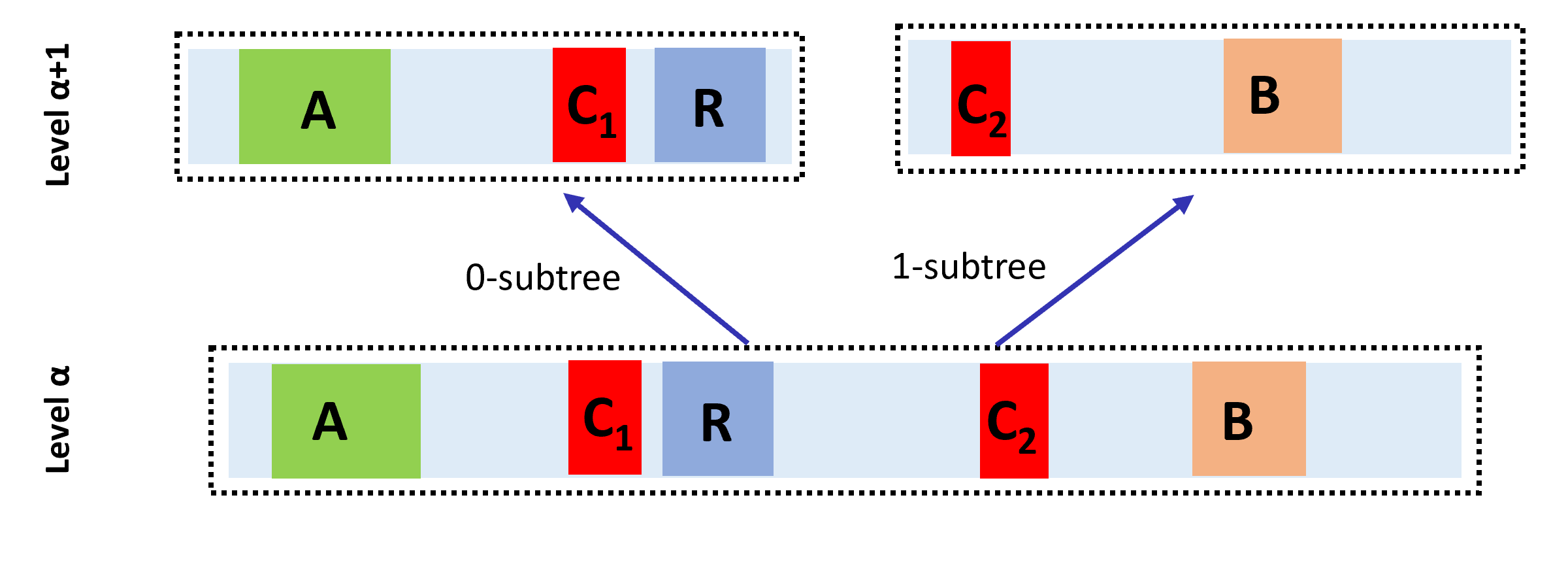}}
\hspace{0.2in}
\centering
    \subfigure[Groups $A$, $C_1$ and $R$ move next to each other without violating Invariant I. Similarly groups $B$ and $C_2$ move next to each other in the 1-subtree at level $\alpha + 1$.]
    {\includegraphics[width=0.45\columnwidth]{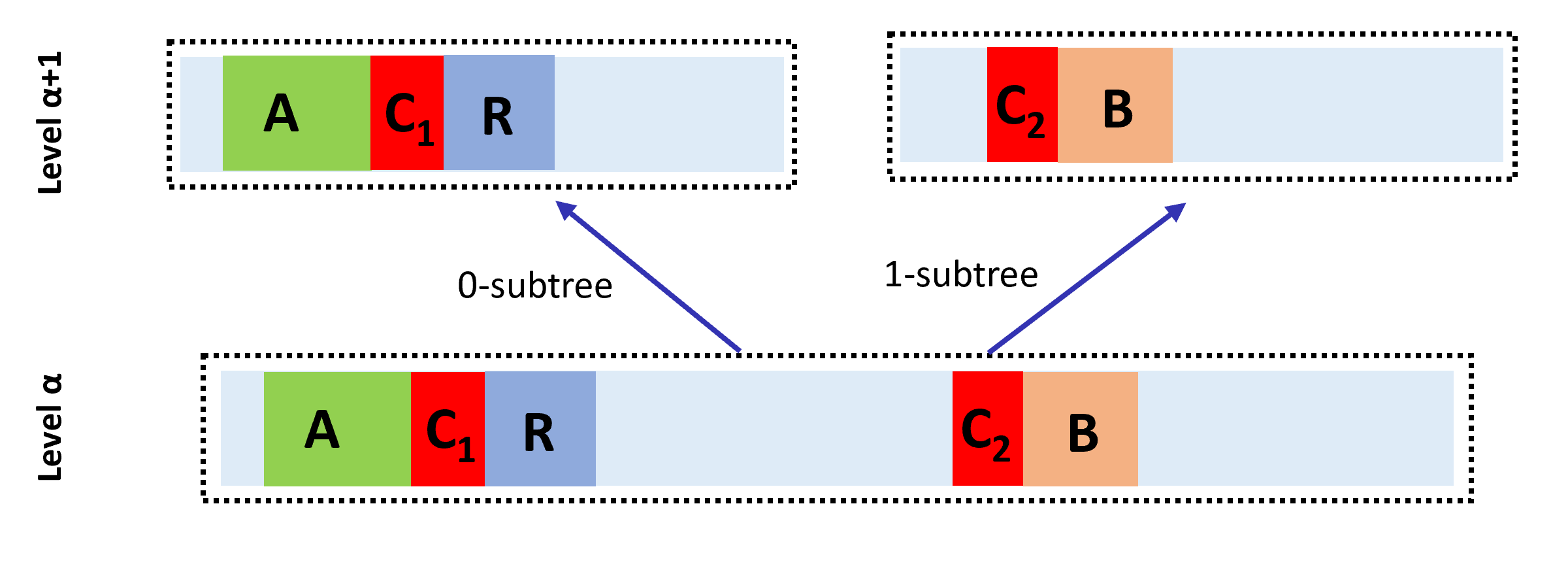}}
\hspace{0.2in}
\centering
    \subfigure[Groups $A$ and $B$ move next to each other without violating Invariant I.]
    {\includegraphics[width=0.45\columnwidth]{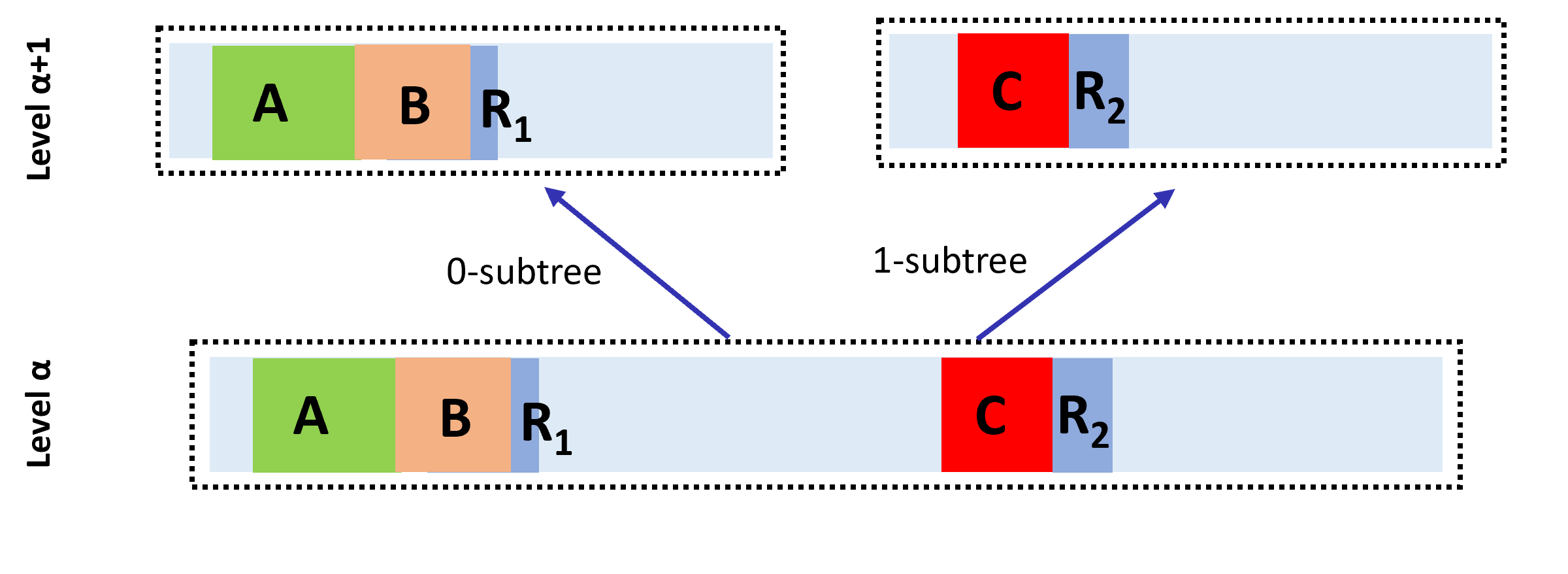}}
\caption{Intergroup transformation example: group $B$ moves next to group $A$, and the size of group $B$ is smaller than the size of the $\alpha$-relative group $C_1$.}
\label{fig:ex_intergroup_random}
\end{figure}

In order to preserve the Invariant I, this procedure requires following cautions:

\begin{enumerate}

\item All the groups involved, including groups $A$, $B$, $C_1$, $C_2$, and groups in portion $R$, cannot be $\alpha+1$-relative. If any of these groups are $\alpha+1$-relative, algorithm \RSG~ moves all the $\alpha+1$-relative groups next to each other recursively.

\item As shown in Figures \ref{fig:ex_intergroup} and \ref{fig:ex_intergroup_random}, moving groups $A$ and $B$ next to each other, may require moving other groups (for example $A$ and $C_1$) next to each other in subtrees $s^u_{\alpha+1}$ and $s^v_{\alpha+1}$. Similarly, these moves may require move moves at the subtrees of upper levels. All these moves take place recursively and parallelly before groups $A$ and $B$ moves next to each other.

\end{enumerate}

Once the groups $A$ and $B$ moves next to each other, they merge into a single group. To merge, all the nodes $x$ of both groups $A$ and $B$ update their necessary group-ids.

\subsubsection{Intra-Group Transformation} \label{sec:intra-group}
The intra-group transformation moves the communicating nodes in a subtree of size 2, as required by our self-adjusting model. Let $\alpha$ be the level of the smallest common subtree that contains both nodes $u$ and $v$ at the beginning of this step. Subtree leap and inter-group transformation ensure that the communicating groups belong to the same group at level $\alpha$. Obviously, the group of nodes $u$ and $v$ at level $\alpha$ is divided into subgroups at upper levels unless nodes $u$ and $v$ are already in a subtree of size 2. Intra-group transformation rearranges these subgroups at different levels using the timestamps stored by the communicating nodes to move nodes $u$ and $v$ in a subtree of size 2.

\textbf{Group-id and Timestamps.} At any time $t$, we require every node to store a \emph{group-id} and two \emph{timestamp}s for each of the $N$ levels. We denote the group-id of node $x$ for level $d$ as $G^x_d$. Each node of a group at some level holds the same group-id for that level, and the timestamps of a node for some level indicates the node's attachment to its group at that level.

Each node has two timestamps for each level, referred to as the T-timestamp and K-timestamp. We denote the T-timestamp and K-timestamp of a node $x$ for a level $d$ as $T^x_d$ and $K^x_d$, respectively. Initially, for any node $x$, $T^x_N = K^x_N = \infty$, and $T^x_i = K^x_i = 0$ for any level $i < N$, where $N$ is the dimension of the entire hypercubic network.

For any node $x$ and level $d$, the algorithm ensures that most of the nodes in the group of node $x$ at level $d$ are likely to be in the connected component $\mathcal{G}_x(T^x_d)$. Also, for any node $x$ and level $d$, the algorithm ensures that most of the nodes in the group of node $x$ at level $d+1$ and in the subtree $\sim s^x_{d+1}$ are likely to be in the connected component $\mathcal{G}_x(K^x_d)$. It is important to note that all the nodes $x$ of a group at level $d$ has the exact same value set to their T-timestamps $T^x_d$. Also, all the nodes $x$ of a group at level $d$ has the exact same value set to their K-timestamps $K^x_{d-1}$. Figure \ref{fig:ex_intragroup} shows an example of intra-group transformation after communication between nodes $A$ and $B$ to give readers a sense of how these timestamps play a role in the rearrangement of the groups.

\begin{figure}[htbp]
\hspace{-0.2in}
\def \subfigcapskip{0pt}
\centering
    \subfigure[A communication graph where the numbers on the edges indicate the time of the most recent communication between the node pair.]
    {\includegraphics[width=0.65\columnwidth]{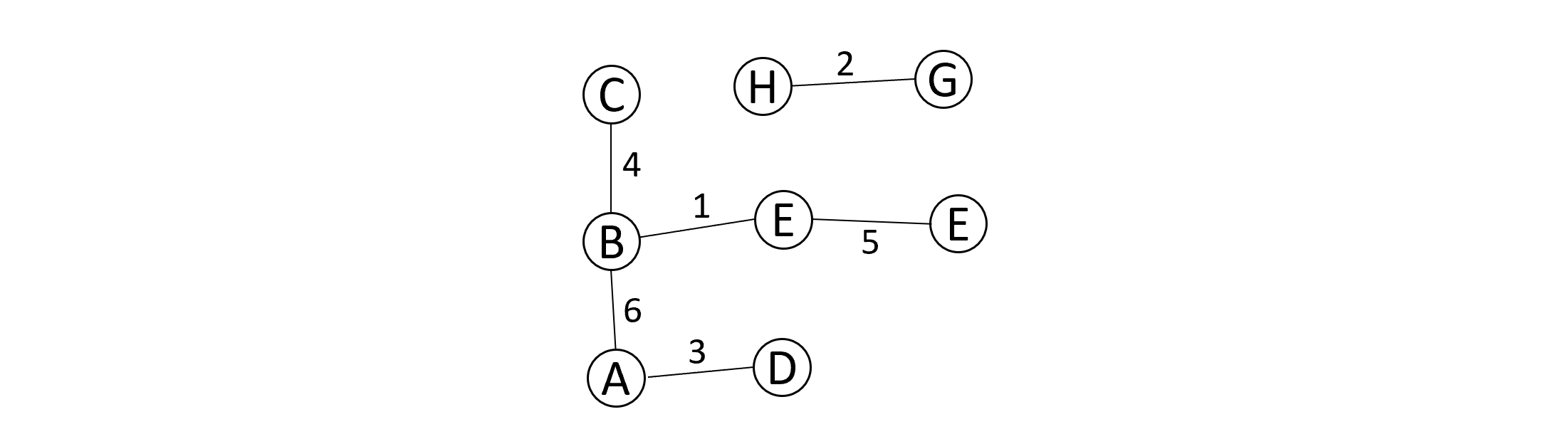}}
\hspace{0.2in}
\centering
    \subfigure[Tree modeling of the network at time 6 ($\mathcal{N}_6$), before the communication $(A,B)$.]
    {\includegraphics[width=0.65\columnwidth]{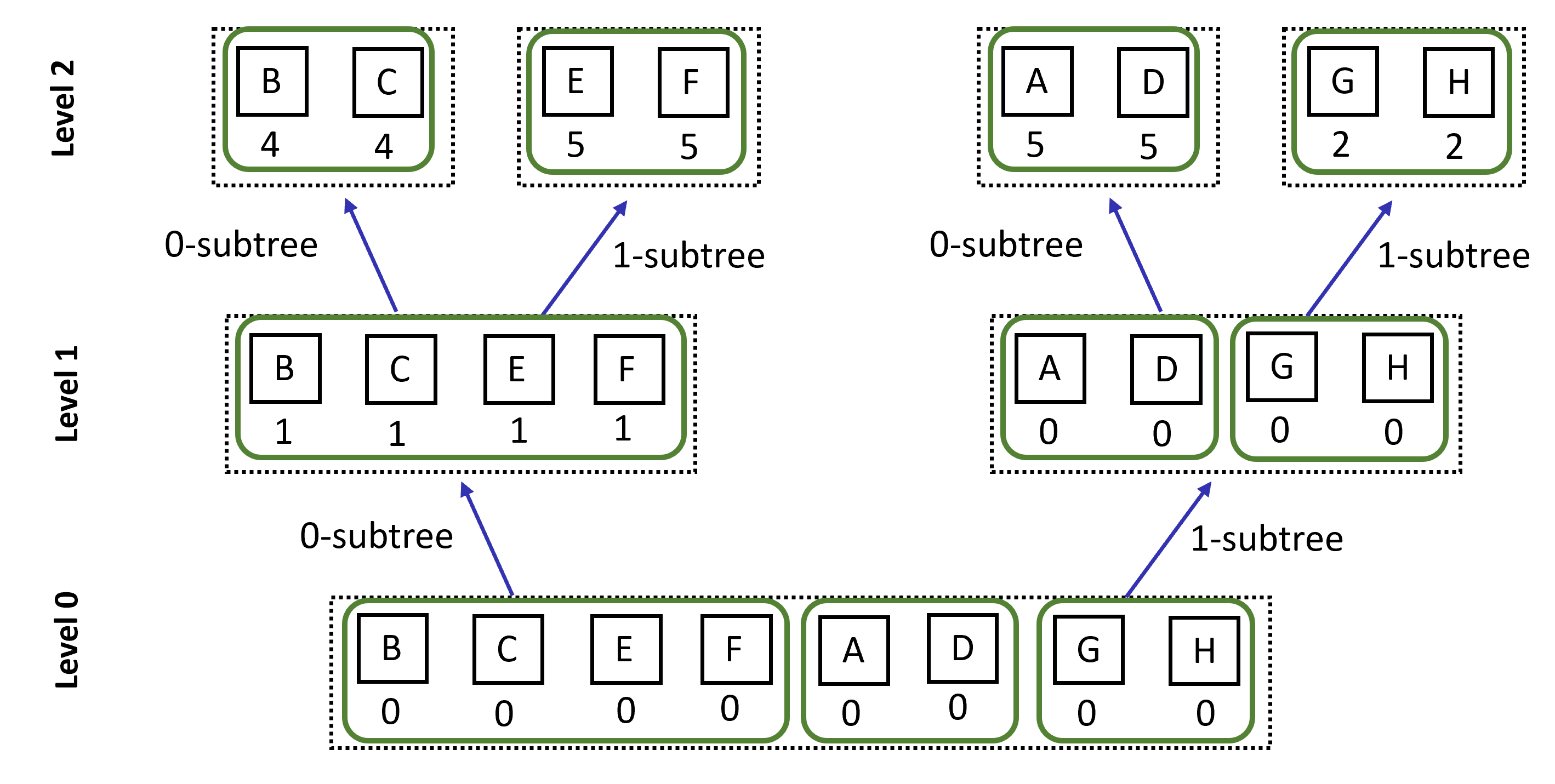}}
\hspace{0.2in}
\centering
    \subfigure[Tree modeling of the network after the transformation following the communication $(A,B)$.]
    {\includegraphics[width=0.65\columnwidth]{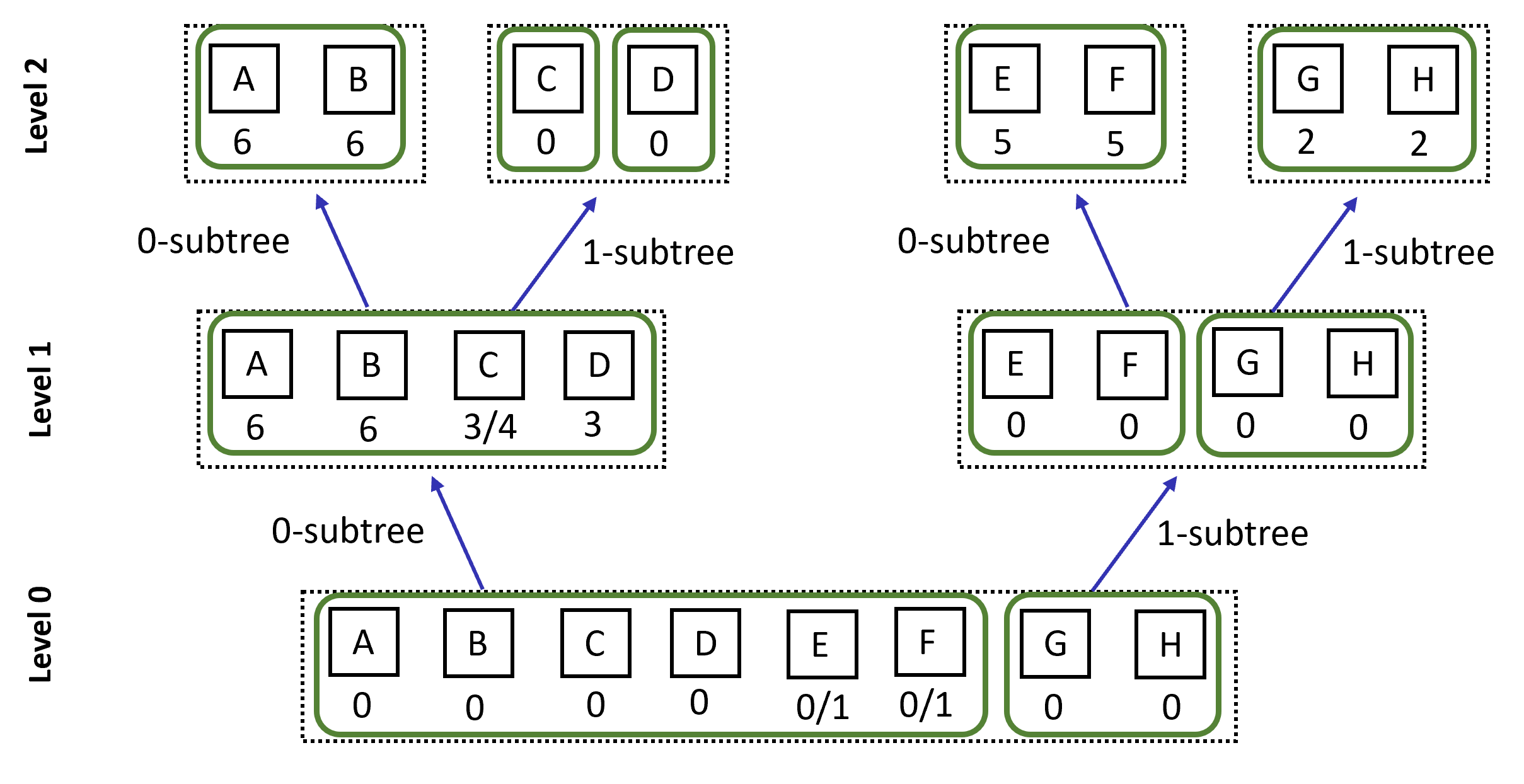}}
\caption{ An example of an intra-group transformation for communication $(A,B)$. The green boxes indicate groups at different levels and the single number below the nodes indicate the node's T-timestamp and K-timestamp for the corresponding level when both timestamps are the same. When the timestamps are different, they are shown in format $T^x_d/K^x_d$.}
\label{fig:ex_intragroup}
\end{figure}

\textbf{Transformation Procedure.}
Given that the communicating nodes $u$ and $v$ are in the same group at level $\alpha$ (but not in level $\alpha + 1$), we compare the size of node $u$'s and $v$'s level-$(\alpha+1)$ groups. Suppose, $group(u,\alpha+1)$ and $group(v, \alpha+1)$ denote the level-$(\alpha+1)$ groups of nodes $u$ and $v$, respectively. Suppose, $|group(u,\alpha+1)| \geq |group(v,\alpha+1)|$  (the algorithm is symmetrical otherwise). We refer to $group(u, \alpha+1)$ and $group(v, \alpha+1)$ as the \emph{dominant} and the \emph{submissive} group, respectively. We count the number of nodes $x \in group(v,\alpha+1)$ such that $K^x_{\alpha} \geq T^y_i$ for each $i >\alpha$ and each $y \in \{u,v\}$. Let $count_y(i)$ be the number of nodes in $group(v,\alpha+1)$ such that $K^x_{\alpha} \geq T^y_i$.

For each $i > \alpha$ such that  $T^u_{i+1} > 0$, we select a set of nodes $S_i$ for repositioning from each subtree $\sim s^u_{i+1}$ by using the following rules:

\begin{itemize}

\item If $count_u(i) \geq |\sim s^u_{i+1}|$, all the nodes of $\sim s^u_{i+1}$ are included in $S_i$.

\item If $count_u(i) < |\sim s^u_{i+1}|$, $S_i$ includes all the $j$-relative groups (if exists) in subtree $\sim s^u_{i+1}$, such that $j\leq i$. If there are any $k$-relative groups of these $j$-relative groups such that $k>i$, then all these $k$-relative groups are also included in $S_i$. Let $R$ be the total number of nodes in all these $j$ and $k$-related groups. If $R < count_u(i)$, then a random coordinate $C$ is chosen in subtree $\sim s^u_{i+1}$ such that $C$ is not the coordinate of any of the nodes in the $j$ and $k$-related groups included in $S_i$. We include the nodes within the coordinate ranged from $C$ to $C + (count_u(i) - R) - 1$ and all the member of their own and relative groups in subtree $\sim s^u_{i+1}$. We move all the nodes in $S_i$ together without violating invariant $I$ using transformation similar to the inter-group transformation. The details about the distributed implementation of moving these nodes together are presented in Section \ref{sec:dist_impl}.


\end{itemize}

For each $i > \alpha$ such that $T^u_{i+1} = 0$ but $T^u_{i} > 0$, we select the set of nodes $S_i$ by using rules similar to above, except that instead of choosing the nodes from subtree $\sim s^u_{i+1}$, we choose the nodes of $S_i$ from the set of nodes $\{x| (x \in s^u_i) \land (x \not \in group(u,i+1))\}$. Then we move together all the nodes of $S_i$ and $group(u,i+1)$ without violating the invariant $I$. Also, set $S_i$ is empty if both $T^u_{i+1}$ and $T^u_{i}$ are zero.

Let $N^\prime$ be the intermediate logical network topology after inter-group transformation and moving together nodes of $S_i$s for each $i > \alpha$. Let $S$ be the set of all the nodes selected for repositioning, which implies $S = \{S_N \cup S_{N-1} \cup \cdots \cup S_{\alpha+1} \cup group(v,\alpha+1)\}$. From this point forward, we reposition the nodes in $S$ to transform the intermediate network $N^\prime$ to $\mathcal{N}_{t+1}$ and the coordinate of any node $x \not \in S$ does not change anymore.

Let $rank(x)$ be the rank of the coordinate of any node $x \in S$ in network $N_{t+1}$ (after reposition) such that, for any two nodes $a$ and $b$ in $S$, $rank(a) < rank(b)$ if and only if (I) tree distance $d_{Tree}(N_{t+1},(u,a)) \leq d_{Tree}(N_{t+1},(u,b))$, and (II) $coord(a) < coord(b)$ in network $\mathcal{N}_{t+1}$ if $d_{Tree}(N^\prime,(u,a)) = d_{Tree}(N^\prime,(u,b))$.

Let $T = \{T_1, T_2, \cdots T_{2(N-\alpha)}\}$ be the sorted (in non-increasing order) list of T-timestamps $T^y_i$ for all $i>\alpha$ and $y \in \{u,v\}$. For nodes in $S_i$ for each $i>\alpha$, we count the number of nodes $x \in S_i$ such that $K^x_i \geq T^y_i$ for each $i >\alpha$ and each $y \in \{u,v\}$. Adding these counts with the previous counts ($count_y(i)$s), we determine $COUNT(i)$ for for each $i, 1 \leq i \leq 2(N-\alpha)$, where $COUNT_(i)$ denotes the number of nodes $x \in S$ with $K^x_i > T_i$, and $T_i \in T$.

Let for any node $x \in S$, $k(x)$ be the lowest integer such that $K^x_j \geq T_{k(x)}$, where $j$ is an integer such that the tree distance from node $x$ to the nearest communicating node ($u$ or $v$) in network $\mathcal{N}_t$ is $N-j$. The reposition takes places using the following rules:

\begin{itemize}
\item In network $\mathcal{N}_{t+1}$, $rank(x) \leq COUNT(k(x))$.
\item Let $group(x)$ denote the level-$(j+1)$ group of node $x$ in network $\mathcal{N}_t$, combined with all its $d$-relative groups that are included in set $S$. All the nodes of $group(x)$ have consecutive ranks in network $\mathcal{N}_t$.

\end{itemize}

Clearly, node $u$ is not included in set $S$. Since $K^v_i$ is set to $\infty$ for any $i>\alpha$, $v$ moves to a subtree $s^u_N$ after the repositioning described above.

\textbf{Timestamps update:} \RSG~ has the following timestamp rules:
\begin{itemize}
\item [\textbf{T1}] Each node $x$ has a counter $C^x_d$ and a next-T-timestamp variable $nextT^x_d$ for each level $d$. All the counters are initialized to zero at the beginning. Let $(u,v)$ be a pair of communicating nodes such that $v \in S$ in the intra-group transformation. Let $j_x$ be an integer such that the tree distance between node $x$ and the nearest communicating node in intermediate network $N^\prime$ is $N-j_x$. Let $X_i$ be the set of nodes such that $X_i \subseteq S$ and for each node $x \in X_i$, $T_{i-1} < K^x_{j_x} \leq T_i$, where $1 < i < \alpha$.

    Now, for each level $d > \alpha$, we perform the following computation. If there exists an integer $i$ and node pair $\{a,b\} \subseteq X_i$, such that $a \in s^u_d$ and $b \in s^u_{d-1}$ in network $\mathcal{N}_{t+1}$, we compute $k$, where $k$ is the number of nodes in $X_i$ that are in subtree $s^u_d$ in network $\mathcal{N}_{t+1}$. We compute the approximate $L$-th largest $K^x_{j_x}$ among all nodes $x \in X_i$, where $L = \bigg( \Big( \left\lceil \frac{k}{N} \right\rceil + 1 \Big)  \frac{2^{\left\lceil \log |X_i| \right\rceil}}{N} \bigg)$. In section \ref{sec:dist_impl}, we present a distributed algorithm the computation of $k$ and the approximate $L$-th largest $K^x_{j_x}$ in $O(\log X_i)$ rounds with $O(X_i)$ messages. We set the value of $nextT^x_d$ as the approximate $L$-th largest timestamp in $X_i$. Then we move together all the nodes $x$ in $X_i$ with $K^x_{j_x}$ larger or equal to the approximate $L$-th largest timestamp, and place $k$ of them in subtree $\sim s^u_d$ in network $\mathcal{N}_{t+1}$. This ensures that all the $k$ nodes that are placed in subtree $\sim s^u_d$ has $K^x_{j_x} \geq nextT^x_d$.

    All nodes $x \in s^u_d$ update their T-timestamp as $T^x_d \leftarrow nextT^x_d$ if $C^x_d + k \geq |s^u_{d+1}|$, and update the counter as $C^x_d = (C^x_d + k) \mod |s^u_{d+1}|$. In other words, the T-timestamp is updated to a new value when the algorithm believes that there are sufficient number of nodes in subtree $s^u_d$ that are in the connected component $\mathcal{G}_u (nextT^x_d)$.

\item [\textbf{T2}] Any node moves from subtree $\sim s^u_d$ to subtree $\sim s^u_{d^\prime}$, where $d$ and $d^\prime$ are levels such that $\alpha \leq d^\prime < d$, update their K-timestamps as $K^x_{d^\prime} \leftarrow K^x_d$ and $K^x_d \leftarrow 0$.

\item [\textbf{T3}] Both the communicating nodes set their T-timestamp and K-timestamp for level $N-1$ to $t$.

\end{itemize}


\subsection{Distributed Implementation} \label{sec:dist_impl} 

Upon a communication request $(u,v)$ and routing at time $t$, nodes $u$ and $v$ share (between each other) their group-ids, T-timestamps, start and end coordinates of their groups and related groups for levels $\alpha, \alpha+1, \cdots N-1$. The general idea is that nodes $u$ and $v$ send these information to the nodes that will take part in the transformation (i.e. selected for repositioning), and then the nodes locally compute their new coordinate in network $\mathcal{N}_{t+1}$, using the broadcasted and local information. Note that, this requires at most $N-\alpha$ rounds (logarithmic to working set number) and at most $2^{N-\alpha}$ messages (linear to working set number), and the nodes partially simulate the transformation to compute their new coordinate in parallel. The nodes then communicate with the node in their target coordinate (in parallel) and simply acquires the new links to complete the transformation. Communicating with target nodes requires at most $N-\alpha$ rounds (logarithmic to working set number) and amortized $2^{N-\alpha}$ messages (linear to working set number).

All the nodes in the network has a pseudo random number generator. Upon the communication $(u,v)$, nodes $u$ and $v$ agree on a random seed. Let $l(x)$ be the lowest level such that a pair of $l(x)$-relative groups exist in subtree $s^x_{l(x)}$. Each communicating node $x$, $x \in \{u,v\}$, broadcasts the group-ids, T-timestamps, start and end coordinates of groups and related groups of both communicating nodes for levels $\alpha, \alpha+1, \cdots N-1$, and the random seed to all nodes in subtrees $s^x_{\min(l(u),l(v),\alpha+1)}$.

If $\min(l(u),l(v)) > \alpha + 1$, subtree leap is required and subtree $s^v_{\min(l(u),l(v))}$ is swapped by subtree $\sim s^u_{\min(l(u),l(v))}$. Each node $x$ in subtree $s^v_{\min(l(u),l(v))}$ computes their target coordinate after subtree leap as the concatenation of first $\min(l(u),l(v)) -1$ bits of the coordinate of $u$ and last $N-\min(l(u),l(v))$ bit of the coordinate of node $x$.

Both inter-group and intra-group transformations require moving a set of groups together without violating the Invariant I. Given that all the nodes involved in any inter-group or intra-group transformation have the start and end coordinates of the related groups of nodes $u$, $v$ and their own, nodes simulate their new position in rounds. In the first round, each node $x$ exchanges the start and end coordinates of their level-$(d+1)$ groups with the node $y$ such that the coordinates of nodes $x$ and $y$ differ only by the last bit. With these information, nodes can simulate if any move is necessary in subtree $s^x_{N-1}$ and compute the new start and end coordinates of their groups. In the second round, nodes $x$ forward their simulated start and end coordinates and $K^x_{d+1}$ to the node $y$ such that the coordinates of nodes $x$ and $y$ differ only by the second last bit. Similarly, nodes simulate any necessary move in subtree $s^x_{N-2}$. This continues recursively and parallelly until all the movements are simulated. The number of rounds required for this simulation is logarithmic of the number of nodes involved in the simulation.

The necessary counts can be computed during the process of simulation of the movements. For computing the $L$-th smallest approximate K-timestamp, we also use a similar technique as described above. For computing the $L$-th smallest approximate K-timestamp in a set of $k$ nodes, we construct a temporary hypercube of $\left \lceil \log k \right \rceil$ dimensions, such that (a) nodes are divided into subtrees as evenly as possible, and (b) if a subtree of size 2 has only one node (in case $k$ is not a power of 2), place a dummy node to fill the gap. Clearly a subtree of size 2 can have at most one dummy node, and we assign each dummy node exactly the same K-timestamp value of the other node in their subtree of size 2. After the construction, the computation of the approximate timestamp is similar to the Approximation Median Finding (AMF) algorithm we proposed in \cite{DSG}. Note that, AMF uses a temporary skip list instead of a temporary hypercube. This computation requires $O(\log k)$ rounds and it conforms to the $\mathcal{CONGEST}$ model of communication.

We summarize \RSG~ in Algorithm \ref{alg:rsg}.

\begin{algorithm*}[htb]

\DontPrintSemicolon
\caption{\RSG~}
\label{alg:rsg}

Upon a communication request between nodes $u$ and $v$ at time $t$:\;

Establish communication using the standard routing algorithm of the network and perform subtree leap if necessary.

\If{$G^u_\alpha \neq G^v_\alpha$}{
    Perform inter-group transformation to merge the groups of nodes $u$'s and $v$'s at level $\alpha$, where $\alpha$ is the highest level such that subtree $s^u_\alpha$ contains both nodes $u$ and $v$ in the intermediate network after subtree leaps and inter-group transformation.
}

Nodes in the level-$(\alpha+1)$ group of node $v$, $group(v,\alpha + 1)$ (assuming that $|group(u,\alpha + 1)| \geq |group(v,\alpha + 1)|$), compute the counts $count_y(i)$.

Node $v$ sends the count $count_u(i)$ to a randomly chosen node $r_i$ in subtrees $\sim s^u_i$ for each $i \in \{ N, N-1, \cdots, \alpha+1 \}$.

Any node $r_i$ receiving the count $count_u(i)$ finds the nodes in set $S_i$ and move them together without violating the Invariant $I$. Let, $N^\prime$ be the intermediate network after nodes in $S_i$ are moved together in subtree $\sim s^u_i$.

During the process of moving together in the previous step, nodes of $S_i$ computes counts similar to counts $count_y(i)$. Node $r_i$ sends the start and end coordinates of the nodes in $S_i$ in network $N^\prime$ and the computed counts to node $v$.

Node $v$ computes the counts $COUNT(i)$ as described in Section \ref{sec:intra-group}. Node $v$ sends these counts to all nodes in set $S = \{S_N \cup S_{N-1} \cup \cdots \cup S_\alpha\}$. Each node $x \in S$ computes their new coordinate in transformed network $\mathcal{N}_{t+1}$ and move to their new position.

All nodes $x$ in subtree $s^u_{\alpha}$ update their timestamps, counters, and group-ids.


\end{algorithm*}

%% file: analysis.tex
\presec
\section{Analysis}
\label{sec:analysis}
\postsec

\begin{lemma}
\label{ts-final-lemma}
\textbf{(Timestamp Lemma)} For any node $u$, any level $d$ such that $T^u_d >0$ at time $t$, the expected number of nodes $x \in (s^u_d \cap \mathcal{G}_u(T^u_d,t))$ is at least $0.63 \cdot |s^u_d|$.
\end{lemma}
\begin{proof}
From the timestamp rule T1, we know that $T^u_d$ is updated  every time the counter $C^u_d$ is incremented by $|\sim s^u_{d+1}|$. This implies that, in between two consecutive updates of $T^u_d$, at least $|\sim s^u_{d+1}|$ nodes in subtree $\sim s^u_{d+1}$ are replaced. Let $X_1, X_2, \cdots, X_k$ be the communications between two consecutive updates of $T^u_d$, where only one of the communicating nodes in $X_i$ is from subtree $s^u_d$. Let $E[Y]$ be the expected number of nodes selected for repositioning in subtree $\sim s^u_d$ by the transformations followed by communications $X_1, X_2, \cdots, X_k$. Using the classic balls and bins model, we can show that:

\begin{equation}
\label{eq:ey}
E[Y] = 1 - \sum_{i=1}^{|\sim s^u_{d+1}|} \prod_{j=1}^{k}(1-p_{i,j})
\end{equation}

where $p_{i,j}$ is the probability that the $i^{th}$ coordinate in subtree $\sim s^u_{d+1}$ is selected for repositioning by the transformation followed by communication $X_j$.

Now, if $k = |\sim s^u_{d+1}|$, $E[Y] = 1 - |\sim s^u_{d+1}|\big(1-\frac{1}{|\sim s^u_{d+1}|}\big)^{|\sim s^u_{d+1}|} \approx 1 - \frac{|\sim s^u_{d+1}|}{e} > 0.63 \cdot |\sim s^u_{d+1}|$. On the other hand, if $k = |\sim s^u_{d+1}|$, more than one nodes must be selected for repositioning by the transformation followed by a communication $X_j$. Which implies, there exists at least a $(i,j)$ pair, for which $p_{i,j} > \frac{1}{|\sim s^u_{d+1}|}$. Thus, from Equation \ref{eq:ey}, we get $E[Y] >  - \frac{|\sim s^u_{d+1}|}{e} > 0.63 \cdot |\sim s^u_{d+1}|$. Clearly, after transformations $X_1, X_2, \cdots, X_k$, subtree $\sim s^u_{d+1}$ will have at least $0.63 \cdot |\sim s^u_{d+1}|$ nodes $x$ from $\mathcal{G}_u(T^u_d,t)$, as for all these nodes $x$, $K^x_d \geq T^u_d$.

We argue that the above argument is also true for all subtrees $\sim s^u_e$ where $e>d$. Now, according to T1, $T^u_e \geq T^u_d$ for all $e>d$. Given that the lemma holds trivially at time 0 for all levels, and also for $d = N - 1$ at any time, the lemma holds for all levels $d$ for all times $t$.
\end{proof}

\begin{lemma}
\label{k-order-lemma}
\textbf{(K-Order Lemma)} Let a communication $(u,v)$ takes place at time $t$ and a node $x$ is repositioned and placed in subtree $\sim s^u_d$ in network $\mathcal{N}_{t+1}$. Let $S_{d+1}$ be the set of nodes repositioned by the intra-group transformation and placed in $\sim s^u_{d+1}$. At least $0.8 \cdot |S_{d+1}|$ nodes from set $S_{d+1}$ have their level-$(d+1)$ K-timestamp greater than $K^x_d$ at time $t+1$.
\end{lemma}
\begin{proof}
Since \RSG~ counts nodes with K-timestamps higher than all the T-timestamps of the involved levels of both the communicating nodes, it is easy to see that among the nodes that are placed in subtree $s^u_d$ by intra-group transformation, no more than $|S_{d+1}|$ nodes can have K-timestamps from the same range of T-timestamps as counted by \RSG. The approximate $L$-th largest K-timestamp computation ensures that at least $\frac{N-1}{N}$ nodes in $S_{d+1}$ have their level-$(d+1)$ K-timestamp greater than $K^x_d$ at time $t+1$, where $N$ is the number of dimensions in the hypercube network. If $N>4$, the lemma holds directly. Otherwise, the subtree $\sim s^u_{d+1}$ can have at most 4 nodes, and all the nodes in $S_{d+1}$ have their level-$(d+1)$ K-timestamp greater than $K^x_d$.

\end{proof}

\begin{lemma}
\label{lem:exp-dist}
\textbf{(Intra-Group Distance Lemma)} Let a node pair $(u,v)$ communicated at time $t^\prime$ for the last time prior to time $t$. Let $\alpha_t$ be the highest level such that $s^u_\alpha$ is a common subtree for both nodes $u$ and $v$ in network $\mathcal{N}_t$. If nodes $u$ and $v$ belong to the same group at level $\alpha$, the expected tree distance between nodes $u$ and $v$ in network $\mathcal{N}_t$ is at most $\log T_t(u,v) + 1$.
\end{lemma}

\begin{proof}
Let $g(t)$ denote node $u$'s (and also node $v$'s) group at level $\alpha_t$ in network $\mathcal{N}_t$. Let $t_i$ denote the time at which the $i$-th communication $(x,y)$ took place after time $t^\prime$, such that $x \in g(t_i)$ and $y \not \in g(t_i)$. Let there has been $m$ such communications between time $t^\prime$ and time $t$. Obviously, $t^\prime < t_1 < t_2 < \cdots < t_m < t$.

If nodes $u$ and $v$ are in a submissive group during the intra-group transformation at any time $t_i$, we show that the expected distance between nodes $u$ and $v$ in network $\mathcal{N}_{t_i + 1}$ is at most $\log T_{t_i + 1}(u,v)$, assuming that their tree distance in network $\mathcal{N}_{t_i}$ is at most $\log T_{t_i}(u,v) + 1$. To this end, we analyze two possible cases. In the first case, we assume that the communicating node from the submissive group is not from $g(t_i)$. In this case, all the nodes of $g(t_i)$ will move together in network $\mathcal{N}_{t_i + 1}$, and the tree distance between nodes $u$ and $v$ will not exceed their tree distance in network $\mathcal{N}_{t_i}$. In the second case, we assume that one of the communicating nodes is from $g(t_i)$. Let $x$ be the communicating node from $g(t_i)$. According to Lemma \ref{ts-final-lemma}, in network $\mathcal{N}_{t_i + 1}$, there are at least $0.63 \cdot |s^x_{\alpha_{t_i+1} + 1}|$ nodes from $(s^x_{\alpha_{t_i+1} + 1} \cup \mathcal{G}_x(T^x_{\alpha_{t_i+1} + 1}, t_i+1))$. Thus the lemma follows.

On the other hand, let us consider that the nodes $u$ and $v$ are in a dominant group during the intra-group transformation at any time $t_i$. We show that the expected distance between nodes $u$ and $v$ in network $\mathcal{N}_{t_i + 1}$ is at most $\log T_{t_i + 1}(u,v) + 1$. Let us make a contradictory assumption that $d_{Tree}(\mathcal{N}_{t_i + 1},(u,v)) > \log T_{t_i + 1}(u,v) + 1$. Obviously there exists a node $x \in g(t_i + 1)$ ($x$ may be $u$ or $v$ or a different node) such that timestamp $T^x_i > 0$ for all $i \geq \alpha_{t_i+1}$. Let $x$ be a node in subtree $s^u_{\alpha_{t_i + 1}}$ (the analysis is similar if $x$ is in $s^v_{\alpha_{t_i + 1}}$ instead).


We have three possible cases. In case 1, we assume that $\alpha_{t_i} = \alpha_{t_i-1}$. In this case, the tree distance between nodes $u$ and $v$ does not increase. In case 2, we assume that $\alpha_{t_i} - \alpha_{t_i-1} = 1$. Note that, this can only happen if all the nodes placed in subtree $s^v_{\alpha_{t_i} + 1}$ by intra-group transformation at time $t_i$ has K-timestamp higher than $K^v_{\alpha_{t_i}}$. Let node $v$ be placed in subtree $\sim s^u_{\alpha_{t_i} + 1}$ for the first time after time $t^\prime$ at time $t_j$, where $t_j < t_i$. Let $p_i$ be the probability that node $v$ is selected for repositioning at time $t_i$ and all the nodes placed in subtree $s^v_{\alpha_{t_i} + 1}$ by intra-group transformation at time $t_i$ has K-timestamp higher than $K^v_{\alpha_{t_i}}$. Let $S(i)$ be the number of nodes selected for replacement from subtree $s^v_{\alpha_{t_i}}$ at time $t_i$. According to \RSG:

\begin{equation}
p_i \leq \frac{S(i)}{|s^v_{\alpha_{t_i} + 1}|}
\end{equation}

Let $E[Y]$ be the expected number of times node $v$ is selected for repositioning at time $t_a$, where $t_j < t_a$, given that all the nodes placed in subtree $s^v_{\alpha_{t_i} + 1}$ by intra-group transformation at time $t_a$ has K-timestamp higher than $K^v_{\alpha_{t_i}}$. Since the random selections of coordinate are independent of each other, using \emph{Poisson Binomial Distribution} \cite{Wang93} we get:

\begin{equation}
E[Y] = \sum_{a>j} p_a
\label{eq:pbd}
\end{equation}

Using the similar argument we used to prove Lemma \ref{ts-final-lemma}, we can show that the expected number of nodes selected in subtree $s^v_{\alpha_{t_i}}$ during the time interval $(t_j, t_i]$ is at least $0.63 \cdot |s^v_{\alpha_{t_i}}|$. Using Lemma \ref{k-order-lemma}, at least $0.8 \times 0.63 \cdot |s^v_{\alpha_{t_i}}|$ nodes in subtree $|s^v_{\alpha_{t_i}}|$ have their level-$\alpha_{t_i}$ K-timestamp higher than $K^v_{{t_i}+1}$. Which implies the tree distance between nodes $u$ and $v$ at time $t_i + 1$ is $\log T_{t_i + 1}(u,v) + 1$, and this is a contradiction.

In case 2, we assume that $\alpha_{t_i} - \alpha_{t_i-1} > 1$. We can show a contradiction similar to the one we show for case 2.

\end{proof}

\noindent
\textbf{Definition (\textit{Working Set Distance}).} For a hypercube $\mathcal{N}_i$ at time $i$, the
\emph{working set distance} for any node pair $(x,y)$ is $\left \lceil \log_2 T_i(x,y) \right \rceil$.

\noindent
\textbf{Definition (\textit{Relative Distance}).} For a hypercube $\mathcal{N}_i$ at time $i$, the
\emph{relative distance} of a $k$-relative group with $n$ nodes is $k - \left \lceil \log_2 n \right \rceil$.

\begin{theorem}

\textbf{(Routing Theorem)} For any communication sequence $\sigma = ( \sigma_1, \sigma_2, \cdots \sigma_m)$, the expected routing cost for \RSG~ is at most $2 (WS(\sigma) + m)$.
\label{theorem:hyp_routing}
\end{theorem}

\noindent \emph{Proof Idea.} According to Invariant I, there can be at most $2^i$ $(\log n - i)$-relative groups. Since relative groups are chosen randomly, the adversary needs to try communications in different groups to find a communication between a relative groups. Since we charge any communication twice more than the its corresponding working set distance plus one, the extra charge is used to pay for communications between relative group pairs.

\begin{proof}
Let $s_d$ be a subtree of height $N-d$ and $s_d$ has only one relative group $g$. Obviously, $g$ is a $k$-relative group where $k<d$. Due to the randomness in positioning of nodes within a group, it is easy to see that in order to request a communication between a node from $g$ and its relative group, the adversary needs to request one or more communications with total working set distance at least $N-d$. Since we change every communication an extra of working set distance plus one, we say the \emph{potential} of subtree $s_d$, $P(s_d)$, is $N-d+1$.

For any network of size $n$ (i.e. $2^N = n$), the sum of relative distances for all relative groups is maximum when every subtree of height 2 has exactly one relative group. hence, the total relative distance is at most $n \sum_{i=2}^{\log_2 n} \frac{i}{2^i} < \frac{3n}{2}$. We show that the total potential of is never less that the half of the total relative distance of the network. When every subtree of height 2 has exactly one relative group, the total potential is $\frac{3n}{4}$. Now if there is exactly one subtree of height $k$ has exactly 1 relative group (all the subtrees of height 2 in the remaining network still has exactly 1 relative group), the total relative distance will reduce by at least $\frac{2\times 2^k}{4} = 2^{k-1}$ and total potential will be reduced by $2^{k-1} - k - 1$. With this argument, we can show that regardless of how many relative groups we have in the network, total potential is never less that the half of the total relative distance of the network.

For every pair of relative groups, we counted relative distance for both groups. Since the positioning of the subgroups and nodes in a group are randomized, it is easy to see that the expected additional distance (i.e. distance beyond working set distance) for a communication between a pair of relative groups is at least half of their combined relative distance. Thus, the expected additional cost for communications between relative group pairs is less than or equal to the total potential of the network at all time.

From Lemma \ref{lem:exp-dist}, we know that the expected cost for any communication $(u,v)$ at time $t$, where nodes $u$ and $v$ are not separated by related groups is $\log T_i(u,v) + 1$. We charge each communication twice of that to cover the additional costs for communications between related group pairs. Hence the expected routing cost for any sequence $\sigma$ of $m$ communications is at most $2(WS(\sigma) + m)$.

\end{proof}

Also, according to the construction of the algorithm, the transformation cost is exactly same as the communication cost.

\begin{theorem}
\textbf{Optimality Theorem.} For any communication sequence $\sigma = ( \sigma_1, \sigma_2, \cdots \sigma_m)$, the expected total cost (i.e. routing and transformation cost) for \RSG~ is at most $O(\log \log n)$ factor more than that of the optimal algorithm.
\end{theorem}
\begin{proof}
This directly follows from Theorems \ref{theorem:hyp_routing} and \ref{theorem:hypercube_ws}.
\end{proof}

%% file: rsg_server.tex
\presec
\section{\RSGS: Algorithm for the single-server model} \label{sec:rsgs}
\postsec

We propose a simple algorithm \RSGS~ for the single-server model. In this model, we assume that the network has a single server $u$, and all other nodes communicate with node $u$. For any communication $(u,v)$, we move node $v$ in subtree $s^u_{N-1}$ (i.e. node $u$'s subtree of size 2) such that nodes $u$ and $v$ get directly connected in subtree $s^u_{N-1}$.

In \RSGS, nodes are not required to maintain any variables, given that the coordinate of the server is known to all the nodes. Also, the message complexity is $O(\log T_i (u,v))$ for \RSGS, where the amortized message complexity for \RSG~ is $O(T_i (u,v))$.

We present the algorithm \RSGS~ in Algorithm \ref{alg:rsgs}.

\begin{algorithm*}[htbp]

\DontPrintSemicolon
\caption{\RSGS~}
\label{alg:rsgs}

Upon a communication request between nodes $u$ and $v$ at time $t$:\;

Establish communication using the standard routing algorithm and record $\alpha$.

\For{level $d = \alpha+1, \alpha+3, \cdots, N-1$ (sequentially)}{
    Find a node $r$ uniformly at random from subtree $\sim s^u_d$\;
    Swap node $v$ with $r$\;
}


\end{algorithm*}


For a unknown sequence of communications $\sigma = \sigma_1, \sigma_2, \cdots, \sigma_m$, let the communication $\sigma_i$ take place at time $i$. To understand the behavior of algorithm \RSGS, we present the following lemma.

\begin{lemma}
\label{one-node-ts-lemma}
\textbf{(Single-Server Time Lemma)} Let $V_t$ be the set of vertices in the communication graph $\mathcal{G}_u(t-2^{N-d}-1,t)$. For any level $d$ and any time $t \geq t(2^{N-d}-1)$, the expected number of nodes $x \in (s^u_d \cap V_t)$ is at least $0.72 \cdot |s^u_d|$.
\end{lemma}

Due to the length of the proof, we preset the proof in Appendix \ref{sec:one_node_lemma_proof}.

\begin{theorem}

\textbf{(Single-Server Theorem)} For any communication sequence $\sigma = ( \sigma_1, \sigma_2, \cdots \sigma_m)$, the expected routing cost for \RSGS~ is at most $WS(\sigma) + m$.
\label{theorem:rsgs}
\end{theorem}

%% file: conclusions.tex
\presec
\section{Conclusion} \label{sec:conclusion}
\postsec

We present a self-adjusting algorithm for hypercubic networks that relies on randomization and grouping of frequently communicating nodes at different levels. The transformation cost of algorithm \RSG~ is less than that of algorithm DSG. Although skip graphs have a different network structure, it might be possible to use a randomization technique (similar to the one used in \RSG~) to improve the transformation cost for skip graphs. Despite higher average distance between nodes, the structural flexibility of skip graphs may be utilized in designing a faster self-adjusting algorithm what uses a similar randomization technique used in this chapter. We leave this idea for future work. 

%% file: appen.tex
\section{Appendix} \label{sec:appen}

\subsection{Notation Table}

We summarize the notations used in the paper in Table \ref{tab:notation}.

\begin{table}[htbp]
\scriptsize
\begin{tabular}{|c|l|}
\hline
\textbf{Notation} & \textbf{Description} \\[1ex]
\hline
$Coord(x)$ & Coordinate of node x \\
$Coord_i(x)$ & $i$-th bit in the coordinate of node x \\
$s^x_d$ & Node x's subtree at level $d$  \\
$\sim s^x_d$ & Node $x$'s complementary subtree at level $d$  \\
$s^{rand}_d \subset s^x_y$ & A subtree at level $d$ randomly chosen from subtree $s^x_y$, where $y < d$. \\
$N_t$ & Network at time t\\
$T_i(x,y)$ & The working set number for node pair $(x,y)$ at time $i$\\
$WS(\sigma)$ & $\sum_{i=1}^{m}$ log$( T_i(\sigma_i))$, where $\sigma = \sigma_1, \sigma_2, \dots \sigma_m$\\
$\sigma_i$ & Communication request $(u_i,v_i)$ at time $i$ \\
$d_{Tree}(\mathcal{N}_t, (x,y))$ & Tree distance between nodes $x$ and $y$ in network $\mathcal{N}_t$ \\
$d(\mathcal{N}_t, (x,y))$ & Routing distance between nodes $x$ and $y$ in network $\mathcal{N}_t$ \\
$\mathcal{G}_x(t^\prime)$ & Node $x$'s connected component in the communication graph drawn for the time interval $[t^\prime,t]$ where $t$ is the current time \\
$G^x_d$ & \textit{group-id} of node x for level d\\
$T^x_d$ & \textit{timestamp} for node x at level d \\
$C^x_d$ & \textit{counter} of node x at level d \\
$S^x_d$ & coordinate for the start (i.e. left most) node of the group \\
$E^x_d$ & coordinate for the end (i.e. right most) node of the group \\
$\mathcal{S}^x_d$ & coordinate of the start node of the $d$-relative group in subtree $s^x_d$ \\
$\sim \mathcal{S}^x_d$ & coordinate of the start node of the $d$-relative group in subtree $\sim s^x_d$ \\
$\mathcal{E}^x_d$ & coordinate of the end node of the $d$-relative group in subtree $s^x_d$ \\
$\sim \mathcal{E}^x_d$ & coordinate of the end node of the $d$-relative group in subtree $\sim s^x_d$ \\

\hline
\end{tabular}
\caption{Notations}
\label{tab:notation}
\end{table}

\subsection{Proof of Lemma \ref{one-node-ts-lemma}}
\label{sec:one_node_lemma_proof}

\textbf{Single-Server Time Lemma (Restated).} Let $V_t$ be the set of vertices in the communication graph $\mathcal{G}_u(t-2^{N-d}-1,t)$. For any level $d$ and any time $t \geq t(2^{N-d}-1)$, the expected number of nodes $x \in (s^u_d \cap V_t)$ is at least $0.72 \cdot |s^u_d|$.

\begin{proof}
Let $t^\prime = t-2^{N-d}-1$ and $v_1, v_2, \cdots$ be the nodes that communicated with node $u$ since time $t^\prime$. In this proof, we denote time $i$ as $t_i$.

For time $t$ and levels $d$ and $d^\prime$, let $S_{(t,d^\prime)} = (s^u_{d^\prime} \cap V_t)$ and $\tilde{S}_{(t,d^\prime)} = (\sim s^u_{d^\prime} \cap V_t)$, where $V_t$ is the set of nodes in the graph $\mathcal{G}_u(t^\prime,t)$.

We use proof by induction. The induction hypotheses are:
\begin{enumerate}
\item[\textbf{[H1]}] $E[S_{(t_i,d+3)}] > 0.999 \cdot |s^u_{d+3}|$
\item[\textbf{[H2]}] $E[\tilde{S}_{(t_i,d+3)}] > 0.99 \cdot |s^u_{d+3}|$
\item[\textbf{[H3]}] $E[\tilde{S}_{(t_i,d+2)}] > 0.86 \cdot |\sim s^u_{d+2}|$
\item[\textbf{[H4]}] $E[\tilde{S}_{(t_i,d+1)}] > 0.52 \cdot |\sim s^u_{d+1}|$
\item[\textbf{[H5]}] $E[\tilde{S}_{(t_i,d)}] > 0.125 \cdot |\sim s^u_{d}|$

\end{enumerate}

Note that, the lemma holds automatically if the induction hypotheses are true, since the summation of the inequalities in hypotheses H0-H4 yields $E[S_{(t_i,d)}] \geq 0.72 \cdot |s^u_d|$.

\textbf{Induction Steps.} We prove each of the hypotheses with following assumptions:

\begin{enumerate}
\item[\textbf{[A1]}] $E[S_{(t_i,d+4)}] > 0.999 \cdot |s^u_{d+4}|$
\item[\textbf{[A2]}] $E[\tilde{S}_{(t_i,d+4)}] > 0.99 \cdot |s^u_{d+4}|$
\item[\textbf{[A3]}] $E[\tilde{S}_{(t_i,d+3)}] > 0.86 \cdot |\sim s^u_{d+3}|$
\item[\textbf{[A4]}] $E[\tilde{S}_{(t_i,d+2)}] > 0.52 \cdot |\sim s^u_{d+2}|$
\item[\textbf{[A5]}] $E[\tilde{S}_{(t_i,d+1)}] > 0.125 \cdot |\sim s^u_{d+1}|$

\end{enumerate}


We define the $i^{th}$ position of a subtree as the position associated with the $i^{th}$ coordinate in the subtree. Let $A(x,y,t_s, t_e)$ be the event that the node in the $x^{th}$ position in $\sim s^u_{y}$ is placed in $\sim s^u_{y-1}$ of the transformed network as a result of a swap performed by an intra-group transformation during the time interval $(t_s,t_e]$. This means, a node from $\sim s^u_{y+1}$ is also placed in the $x$th position in $\sim s^u_{y}$ of the transformed network. According to the lemma setting, each position in subtree $\sim s^u_{y}$ is equally likely to be chosen when a swap operation takes place. This is equivalent to throwing balls in $|\sim s^u_{y}|$ bins, where each ball is thrown into a uniformly random bin, independent of other balls.


Let $x^{d}_{(j,i)}$ be an indicator random variable such that $x^{d}_{(j,i)} = 1$ if the $i^{th}$ position in subtree $\sim s^u_{d}$ is chosen $j$ times by an event $\mathcal{A}(x,d,k^\prime,k)$, where $k^\prime = |s^u_d| - 1$ and $k=|s^u_{d-1}| - 1$; $x^{d}_{(j,i)} = 0$ otherwise. Let $X^{d}_j = x^{d}_{(j,1)} + x^{d}_{(j,2)}, + \cdots, + x^d_{(j,|\sim s^u_{d}|)}$.

\textbf{Proof of H1.} This hypothesis holds trivially if $d \leq N - 4$, as $|s^u_{d+3}| = 2$ and $E[S_{(t_i,d+3)}] = |s^u_{d+3}|$ for $d = N - 4$ and $i \geq 1$. Thus, we prove this hypothesis for $d > N - 4$ and it suffices to show the following:

\begin{itemize}
\item [] \textbf{[H1A]} $E[S_{(t_i,d)}] \geq E[\tilde{S}_{(t_i,d)}]$ for any $d$
\item [] \textbf{[H1B]} $E[\tilde{S}_{(t_i,d+4)}] > 0.999 \cdot |s^u_{d+4}|$
\end{itemize}

Since $(u,v_1)$ is the first communication of node $u$, $E[S_{(t_1,d)}] = 2$ and $E[\tilde{S}_{(t_1,d)}] = 0$ for any $d \geq N -1$. Thus H1B holds for time $t_1$. Now for any communication $(u,x)$ onward such that $x \notin s^u_d$, the node $x$ moves in $s^u_d$ and $x \in s^u_d \cap V_t$. However, the node that moves from $s^u_d$ to $\sim s^u_d$ is not in $\sim s^u_d \cap V$ with a nonzero probability. Therefore, H1A holds.

Now we prove H1B. We compute the expected number of times a position in subtree $\sim s^u_{d+4}$ is chosen at least once during a transformation between time $k^\prime$ and $k$.

\begin{equation*}
\begin{split}
E[X^{d+4}_0] & = \sum_{i = 1}^{|\sim s^u_{d+4}|} E[x^{d+4}_{(j,i)}] = |\sim s^u_{d+4}| \cdot \bigg( \frac{|\sim s^u_{d+4}|-1}{|\sim s^u_{d+4}|} \bigg)^{k- k^\prime} = |\sim s^u_{d+4}| \cdot \bigg( \frac{|\sim s^u_{d+4}|-1}{|\sim s^u_{d+4}|} \bigg)^{8 \cdot |\sim s^u_{d+4}| } \\
&\leq \bigg(\frac{1}{e} \bigg)^8 \cdot |\sim s^u_{d+4}| \bigg[\text{Since}, \big(1-\frac{1}{n} \big)^n \leq \frac{1}{e}, \text{when } n>1\bigg]
\end{split}
\end{equation*}
Thus,

\begin{equation*}
\begin{split}
E[X^{d+4}_{\geq1}] &= |\sim s^u_{d+4}| - E[X^{d+4}_0] > \bigg(1 - \bigg(\frac{1}{e} \bigg)^8 \bigg) \cdot |\sim s^u_{d+4}| = 0.9996 \cdot |\sim s^u_{d+4}|
\end{split}
\end{equation*}

Each position $i$ with $x^{d+4}_{\geq1,i}$ will be occupied by a node from $s^u_{d+4}$. We get:

\begin{equation*}
\begin{split}
E[\tilde{S}_{(t_i,d+4)}] &> 0.999 \times E[X^{d+4}_{\geq1}] + 0.99 \times E[X^{d+4}_{0}]  \approx  0.999 \cdot |s^u_{d+4}|
\end{split}
\end{equation*}

\textbf{Proof of H2.}

\begin{equation*}
\begin{split}
E[X^{d+4}_1] &= \sum_{i = 1}^{|\sim s^u_{d+4}|} E[x^{d+4}_{(j,i)}] = |\sim s^u_{d+3}| \cdot \bigg( \frac{|\sim s^u_{d+4}|-1}{|\sim s^u_{d+4}|} \bigg)^{k- k^\prime -1} \leq \bigg(\frac{1}{e} \bigg)^8 \cdot |\sim s^u_{d+4}|
\end{split}
\end{equation*}

\begin{equation*}
\begin{split}
E[X^{d+4}_{\geq2}] &= |\sim s^u_{d+4}| - E[X^{d+3}_0] - E[X^{d+3}_1]  \geq \bigg(1 - 2 \cdot \bigg(\frac{1}{e} \bigg)^8 \bigg) \cdot |\sim s^u_{d+3}|
\end{split}
\end{equation*}

\begin{equation*}
\begin{split}
E[X^{d+3}_0] = \sum_{x = 1}^{|\sim s^u_{d+3}|} E[A(x,d+3)] \leq \bigg(\frac{1}{e} \bigg)^4 \cdot |\sim s^u_{d+3}|
\end{split}
\end{equation*}

Each position $i$ with $x^{d+3}_{(\geq1,i)}$ will be occupied by a node from $s^u_{d+4}$. We get:

\begin{equation*}
\begin{split}
E[\tilde{S}_{(t_i,d+3)}] & > 0.86 \times E[X^{d+3}_{0}] + 0.99 \times E[X^{d+4}_{\geq 1}] +  0.999 \times (E[X^{d+3}_{\geq 1}]- E[X^{d+4}_{\geq 1}])\\
& > 0.87 \times \bigg(\frac{1}{e} \bigg)^4 \cdot |\sim s^u_{d+3}| + 0.99 \times \bigg( 0.9996 \cdot \frac{|\sim s^u_{d+3}|}{2}  \bigg)  + 0.999 \times \bigg( |\sim s^u_{d+3}| - E[X^{d+3}_{0}] -E[X^{d+4}_{\geq 1}] \bigg)\\
& = \bigg( 0.86 \times 0.0183 + 0.99 \times 0.4998 + 0.999 \times \big( 1 -0.0183 - 0.4998 \big)  \bigg) \cdot |\sim s^u_{d+3}| > 0.99 \cdot |s^u_{d+3}|
\end{split}
\end{equation*}

\textbf{Proof of H3.} Similar to the proof of H2, we get,

\begin{equation*}
\begin{split}
E[\tilde{S}_{(t_i,d+2)}] & > 0.52 \times E[X^{d+2}_{0}] + 0.86 \times E[X^{d+3}_{\geq 1}] +  0.99 \times (E[X^{d+2}_{\geq 1}]- E[X^{d+3}_{\geq 1}])\\
& > 0.52 \times \bigg(\frac{1}{e} \bigg)^2 \cdot |\sim s^u_{d+2}| + 0.86 \times \bigg( \bigg(1 - \frac{1}{e^4} \bigg) \cdot \frac{|\sim s^u_{d+2}|}{2}  \bigg)  \\
&+ 0.99 \times \bigg( |\sim s^u_{d+2}| - E[X^{d+2}_{0}] -E[X^{d+3}_{\geq 1}] \bigg)\\
& = \bigg( 0.52 \times 0.135 + 0.86 \times 0.4908 + 0.99 \times \big( 1 - 0.135 - 0.4908 \big)  \bigg) \cdot |\sim s^u_{d+2}|\\
& > 0.86 \cdot |s^u_{d+3}|
\end{split}
\end{equation*}

\textbf{Proof of H4.} We prove this hypothesis in four steps. We divide the time interval $(k^\prime, k)$ into four equal subintervals, and calculate expectations for each subintervals using the expectations calculated for the previous subinterval.

Let $y^{d}(j,i,t_s, t_e)$ be an indicator random variable such that $y^{d}(j,i,t_s,t_e) = 1$ if the $i^{th}$ position in subtree $\sim s^u_{d}$ is chosen $j$ times by an event $\mathcal{A}(x,d,t_s,t_e)$ after time $t_e$; $y^{d}(j,i,t_s,t_e) = 0$ otherwise. Let $Y^{d}(j,t_s,t_e) = y^{d}(j,1,t_s,t_e) + y^{d}(j,2,t_s,t_e), + \cdots, + y^d(j,|\sim s^u_{d}|,t_s,t_e)$.

\begin{equation*}
\begin{split}
E[\tilde{S}_{(k^\prime + \frac{k-k^\prime}{4} ,d+2)}] &> 0.52 \times  E[Y^{d+2}(0,k^\prime, k^\prime +  \frac{k-k^\prime}{4}) ] + 0.87 \times E[Y^{d+2}(\geq1,k^\prime, k^\prime + \frac{k-k^\prime}{4})]  \\
& >  0.52 \times \frac{1}{e^{0.5}} |s^u_{d+2}|   + 0.87 \times  \bigg( 1 - \frac{1}{e^{0.5}}  \bigg) |s^u_{d+2}| = 0.65 \cdot |s^u_{d+2}|
\end{split}
\end{equation*}

Similarly,

\begin{equation*}
\begin{split}
E[\tilde{S}_{(k^\prime + \frac{k-k^\prime}{2} ,d+2)}] &> 0.65 \times \frac{1}{e^{0.5}} |s^u_{d+2}|  + 0.87 \times  \bigg( 1 - \frac{1}{e^{0.5}}  \bigg) |s^u_{d+2}| =  0.73 \cdot |s^u_{d+2}|
\end{split}
\end{equation*}

\begin{equation*}
\begin{split}
E[\tilde{S}_{(k^\prime + \frac{3(k-k^\prime)}{4} ,d+2)}] &> 0.73 \times \frac{1}{e^{0.5}} |s^u_{d+2}| +  0.87 \times  \bigg( 1 - \frac{1}{e^{0.5}}  \bigg) |s^u_{d+2}| =  0.78 \cdot |s^u_{d+2}|
\end{split}
\end{equation*}

Now, we use the above expectations to calculate $E[\tilde{S}_{(t_i,d+1)}]$ in four steps:

\begin{equation*}
\begin{split}
E[\tilde{S}_{(k^\prime +  \frac{k-k^\prime}{4},d+1)}] & > 0.14 \times E[Y^{d+1}(0,k^\prime, k^\prime +  \frac{k-k^\prime}{4}) ] + 0.53 \times E[Y^{d+2}(1,k^\prime, k^\prime +  \frac{k-k^\prime}{4}) ] \\
&+ 0.86 \times (E[Y^{d+1}(\geq 1,k^\prime, k^\prime +  \frac{k-k^\prime}{4}) ]- E[Y^{d+2}(\geq1,k^\prime, k^\prime +  \frac{k-k^\prime}{4}) ])\\
& > 0.125 \times \frac{1}{e^{0.25}}  \cdot |\sim s^u_{d+1}| + 0.53 \times \bigg( \bigg(1 - \frac{1}{e^{0.5}} \bigg) \cdot \frac{|\sim s^u_{d+1}|}{2}  \bigg)  \\
& + 0.86 \times \bigg( 1 - \frac{1}{e^{0.25}} - \frac{1}{2} \bigg(1 - \frac{1}{e^{0.5}} \bigg) \bigg)  \cdot |\sim s^u_{d+1}| \\
& = \bigg( 0.125 \times 0.778 + 0.52 \times 0.15 + 0.86 \times \big( 1 - 0.778 - 0.15 \big)  \bigg) \cdot |\sim s^u_{d+2}|\\
& > 0.24 \cdot |s^u_{d+1}|
\end{split}
\end{equation*}

\begin{equation*}
\begin{split}
E[\tilde{S}_{(k^\prime +  \frac{k-k^\prime}{2},d+1)}] & >E[\tilde{S}_{(k^\prime + \frac{k-k^\prime}{4},d+1)}] \times E[Y^{d+1}(0,k^\prime +  \frac{k-k^\prime}{4}), k^\prime +  \frac{k-k^\prime}{2}) ] \\
& + E[\tilde{S}_{(k^\prime + \frac{k-k^\prime}{4} ,d+2)}] \times E[Y^{d+2}(1,k^\prime, k^\prime +  \frac{k-k^\prime}{4}) ]  + \\
&0.86 \times \bigg(E[Y^{d+1}\big(\geq 1,k^\prime, k^\prime +  \frac{k-k^\prime}{4}\big) ]- E[Y^{d+2}\big(\geq1,k^\prime, k^\prime +  \frac{k-k^\prime}{4}\big) ]\bigg)\\
& > \bigg( 0.24 \times 0.778 + 0.66 \times 0.15 + 0.86 \times \big( 1 - 0.778 - 0.15 \big)  \bigg) \cdot |\sim s^u_{d+2}| > 0.34 \cdot |s^u_{d+1}|
\end{split}
\end{equation*}

Similarly,

\begin{equation*}
\begin{split}
E[\tilde{S}_{(k^\prime +  \frac{3(k-k^\prime)}{4},d+1)}] & > \bigg( 0.34 \times 0.778 + 0.73 \times 0.15 +  0.86 \times \big( 1 - 0.778 - 0.15 \big)  \bigg) \cdot |\sim s^u_{d+2}|\\
& > 0.44 \cdot |s^u_{d+1}|
\end{split}
\end{equation*}

\begin{equation*}
\begin{split}
E[\tilde{S}_{t_i,d+1)}] & > \bigg( 0.44 \times 0.778 + 0.78 \times 0.15 + 0.86 \times \big( 1 - 0.778 - 0.15 \big)  \bigg) \cdot |\sim s^u_{d+2}|\\
& > 0.52 \cdot |s^u_{d+1}|
\end{split}
\end{equation*}

\textbf{Proof of H5 .} To prove this hypothesis, we use the same four-step approach that we used to prove H4. The expectations calculated for different subintervals in the proof of H4 are used to prove H5. We briefly show the calculations below:

\begin{equation*}
\begin{split}
E[\tilde{S}_{(k^\prime +  \frac{k-k^\prime}{4},d)}] & > \bigg( 0 \times \frac{1}{e^{1/8}} + 0.125 \times \frac{1}{2} \times \big( 1 - \frac{1}{e^{1/4}} \big) + 0.52 \times  \big( 1 - \frac{1}{e^{1/8}} - \frac{1}{2} \times \big( 1 - \frac{1}{e^{1/4}} \big) \big)  \bigg) \cdot |\sim s^u_{d}|\\
& =  \bigg( 0 \times 0.88 + 0.125 \times 0.11 + 0.52 \times  \big( 1 - 0.88 - 0.11  \big) \bigg) \cdot |\sim s^u_{d}| > 0.018 \cdot |s^u_{d}|
\end{split}
\end{equation*}

\begin{equation*}
\begin{split}
E[\tilde{S}_{(k^\prime +  \frac{k-k^\prime}{2},d)}] & >  \bigg( 0.018 \times 0.88 + 0.24 \times 0.11 + 0.66 \times  \big( 1 - 0.88 - 0.11 \big) \bigg) \cdot |\sim s^u_{d}| > 0.04 \cdot |s^u_{d}|
\end{split}
\end{equation*}

\begin{equation*}
\begin{split}
E[\tilde{S}_{(k^\prime +  \frac{3(k-k^\prime)}{4},d)}] & >  \bigg( 0.05 \times 0.88 + 0.34 \times 0.11 + 0.74 \times \big( 1 - 0.88 - 0.11 \big) \bigg) \cdot |\sim s^u_{d}| > 0.08 \cdot |s^u_{d}|
\end{split}
\end{equation*}

\begin{equation*}
\begin{split}
E[\tilde{S}_{(t_i,d)}] & >  \bigg( 0.08 \times 0.88 + 0.24 \times 0.11 +  0.796 \times \big( 1 - 0.88 - 0.11 \big) \bigg) \cdot |\sim s^u_{d}| > 0.125 \cdot |s^u_{d}|
\end{split}
\end{equation*}

\end{proof}

\textbf{Base Case.}  We first show that the lemma holds for any $d > N - 3$. Since \RSGS~ places communicating nodes in a subtree of size 2, clearly for $d = N - 1$, the number of nodes $x \in S_{(t,d)}$ is always 2. Also, for $d = N - 2$, the number of nodes $x \in S_{(t,d)}$ is at least 3 with probability 1. Thus the lemma holds for $d > N - 3$.

Now we show that the induction hypotheses hold for $d = N - 3$. If the hypotheses hold for $t = t_7$, then they will also hold for $t > t_7$ because any node that moves in $s^u_d$ at time $t_i$, $i>7$, is the communicating node $v_i$, and $v_i \in (s^u_d \cap V)$. That leaves us in the necessity to show that the hypotheses hold for $t = t_7$.

Since a subtree at level $d$ is split into two subtrees at level $d+1$, we can write the following:

\begin{equation}
\begin{split}
\label {eq:1}
E[|S_{(t_i,d)}|] = E[|S_{(t_i,d+1)}|] + E[|\tilde{S}_{(t_i,d+1)}|] = E[|S_{(t_i,N -1)}|] + \sum_{j=N - 1}^{d+1} E[|\tilde{S}_{(t_i,j)}|]
\end{split}
\end{equation}

Also, for any subtree $s$, the expected number of nodes in $(s \cap V)$ at time $t_i$ depends on the expected number of nodes in $(s \cap V)$ at $t_{i-1}$, the expected number of nodes in $(s \cap V)$ that move in subtree $s$ at time $t_i$, and the expected number of nodes that leave the subtree $s$ at time $t_i$. We get:

\begin{equation}
\begin{split}
\label {eq:2}
E[|\tilde{S}_{(t_i,d)}|] & = E[|\tilde{S}_{(t_{i-1},d)}|] + \frac{E[|\tilde{S}_{(t_{i-1},d+1)}|]}{|\sim s^u_{d+1}|} -\frac{E[|\tilde{S}_{(t_{i-1},d)}|]}{|\sim s^u_{d+1}|} \\
& = E[|\tilde{S}_{(t_{i-1},d)}|] + \frac{E[|\tilde{S}_{(t_{i-1},d+1)}|]}{2^{N- d+2}} -\frac{E[|\tilde{S}_{(t_{i-1},d)}|]}{2^{N - d-1}}
\end{split}
\end{equation}

As we know, $E[|S_{(t_i,N -1)}|] = 2$ for $i\geq 1$, combining Equations \ref{eq:1} and \ref{eq:2}, we get:

\begin{equation}
\begin{split}
\label {eq:3}
E[|S_{(t_i,d)}|] &= 2 + \sum_{j=N - 1}^{d+1} \bigg(  E[|\tilde{S}_{(t_{i-1},j)}|] + \frac{E[|\tilde{S}_{(t_{i-1},j-1)}|]}{2^{N - j-2}} -  \frac{E[|\tilde{S}_{(t_{i-1},d)}|]}{2^{N - j-1}} \bigg)
\end{split}
\end{equation}

We use Equation \ref{eq:3} to calculate $E(S_{(t_7, N -1)}]$, $E(S_{(t_7, N -2)}]$ and $E(S_{(t_7, N -3)}]$, and show that the hypotheses hold for time $t_7$ and $d = N - 3$. For the time $t_4$ and $d = N -2$,

\[E[|S_{(t_4,N - 2)}|] = 2 + \bigg(  E[|\tilde{S}_{(t_{3},N -1)}|] + \big( E[|\tilde{S}_{(t_{3},N -1)}|] \cdot \frac{1}{1}\big) - \big(E[|\tilde{S}_{(t_{3},N -2)}|] \cdot \frac{1}{2} \big) \bigg) \]

Clearly, $E[|\tilde{S}_{(t_{3},N -1)}| = 1$, $E[|\tilde{S}_{(t_{3},N -1)}|] = 1$, and $E[|\tilde{S}_{(t_{3},N -2)}|] = 1$. Thus,

\[E[|S_{(t_4,N - 2)}|] = 2 + \bigg(  1 + \big(  1 \times 1\big) - \big(1 \times \frac{1}{2} \big) \bigg) = 3.5\]

And,

\[ E[|\tilde{S}_{(t_4,N - 1)}|] = E[|S_{(t_4,N - 2)}|] - 2 = 3.5-2 = 1.5\]

\[E[|\tilde{S}_{(t_{4},N -2)}|] =  E[|\tilde{S}_{(t_{3},N -2)}|] + \big( E[|\tilde{S}_{(t_{3},N -1)}|] \cdot \frac{1}{2}\big) - \big(E[|\tilde{S}_{(t_{3},N -1)}|] \cdot \frac{1}{2^{2}} \big) \bigg) \]

\[\implies E[|\tilde{S}_{(t_{4},N -2)}|] =  0 +  1 \times \frac{1}{2} - 0 \times \frac{1}{4}  = \frac{1}{2}\]

Thus, we get:

\[ E[|S_{(t_5,N - 2)}|]  = 2 +  1.5 + \bigg(1 \times \frac{1}{1} \bigg) - \bigg(1.5\times \frac{1}{2} \bigg) = 3.75 \]

\[ E[|S_{(t_6,N - 2)}|]  = 2 +  1.75 + \bigg( 1 \times \frac{1}{1} \bigg) - \bigg(1.75\times \frac{1}{2} \bigg) = 3.88 \]

\[ E[|S_{(t_7,N - 2)}|]  = 2 + 1.88 + \bigg(1 \times \frac{1}{1} \bigg) - \bigg(1.88\times \frac{1}{2} \bigg) = 3.94 \]

\[ E[|S_{(t_8,N - 2)}|]  = 2 +  3.94 + \bigg( 1 \times \frac{1}{1} \bigg) - \bigg(3.94\times \frac{1}{2} \bigg) = 3.97 \]

Similarly,

\[E[|S_{(t_5,N - 3)}|] = E[|S_{(t_5,N - 2)}|] + \bigg(  E[|\tilde{S}_{(t_{4},N -2)}|] + \big( E[|\tilde{S}_{(t_{4},N -1)}|] \cdot \frac{1}{2}\big) - \big(E[|\tilde{S}_{(t_{4},N -3)}|] \cdot \frac{1}{4} \big) \bigg) \]

\[\implies E[|S_{(t_5,N - 3)}|] = 3.75 +  \frac{1}{2} + \bigg( 1.75 \times \frac{1}{2}\bigg) - \bigg( \frac{1}{2} \times \frac{1}{4} \bigg) = 5.0 \]

Similarly,

\[E[|S_{(t_6,N - 3)}|] = 3.88 +  (5.0-3.88) + \bigg( 1.88 \times \frac{1}{2}\bigg) - \bigg( (5.0-3.88) \times \frac{1}{4} \bigg) = 5.66 \]

\[E[|S_{(t_7,N - 3)}|] = 3.94 +  (5.66-3.94) + \bigg( 1.94 \times \frac{1}{2}\bigg) - \bigg( (5.66-3.94) \times \frac{1}{4} \bigg) = 6.20 \]

\[E[|S_{(t_8,N - 3)}|] = 3.97 +  (6.20-3.97) + \bigg( 1.97 \times \frac{1}{2}\bigg) - \bigg( (6.20-3.97) \times \frac{1}{4} \bigg) = 6.62 \]

We can also calculate $E[|\tilde{S}_{(t_8,N - 3)}|]$ as the following:

\[E[|\tilde{S}_{(t_5,N -3)}|] \leq \frac{1}{2} \times \frac{1}{4} - \frac{1}{2} \times \frac{1}{4} \times{1}{8}  = 0.11\]

\[E[|\tilde{S}_{(t_6,N -3)}|] \leq 0.11 + \bigg((5.0-3.88) \times \frac{1}{4}\bigg) - \bigg(0.11 \times {1}{8} \bigg)  = 0.38\]

\[E[|\tilde{S}_{(t_7,N -3)}|] \leq 0.38 + \bigg((5.66-3.94) \times \frac{1}{4}\bigg) - \bigg(0.38 \times {1}{8} \bigg)  = 0.76\]

\[E[|\tilde{S}_{(t_8,N -3)}|] \leq 0.76 + \bigg((6.20-3.97) \times \frac{1}{4}\bigg) - \bigg(0.76 \times {1}{8} \bigg)  = 1.22\]